\documentclass[11pt]{article}

\usepackage{placeins}
\usepackage{graphicx}
\usepackage{geometry}
\usepackage{amsmath,amssymb,amsfonts,amsbsy,amsthm}
\usepackage{thmtools}

\newtheorem{assumption}{Assumption}
\usepackage{subfigure}
\usepackage[round]{natbib}
\usepackage{array}
\usepackage{float}
\usepackage{booktabs}
\usepackage{longtable}
\usepackage{url}
\usepackage{color}
\usepackage[normalem]{ulem}

\theoremstyle{definition}
\newtheorem{proposition}{Proposition}
\newtheorem{theorem}{Theorem}[section]
\newtheorem{corollary}[theorem]{Corollary}

\renewcommand\thmcontinues[1]{Continued}
\newtheorem{lemma}[theorem]{Lemma}
\newtheorem{remark}[theorem]{Remark}

\newcommand{\ep}{\varepsilon}

\newcommand{\R}{\mathbb{R}}
\usepackage{parskip}

\theoremstyle{definition}

\title{Negative Oil \& Nickel Squeeze: A Feedback Model for Extreme Commodity Futures Prices}
\author{Iosif Zimbidis\thanks{ORFE Department, Princeton University} \and Ronnie Sircar\thanks{\hangindent=1.75em
  \hangafter=1 ORFE Department, Bendheim Center for Finance and Andlinger Center for Energy \& the Environment,
Princeton University}}
\date{\today}

\begin{document}
\maketitle
\begin{abstract}
On April 20, 2020, the May front-month WTI oil futures contract, one day before its expiration date, opened near $\$17/$barrel and dropped far below zero in a single trading day, reaching an intraday low of $-\$40.32$ and settling at $-\$37.63$. Such market behavior was unforeseen at the time. This event, and the March 2022 nickel squeeze, illustrate how constraints on physical delivery create pressure to close futures positions and distort futures prices. We develop a feedback model that connects the resulting price distortion to a roll option, and to the imbalance between delivery-constrained long and short positions. Under lognormal benchmark dynamics, the roll-option value is shown to satisfy a pricing PDE which is nonlinear because its payoff depends on the observed futures price, which itself includes the feedback correction. Remarkably, in both cases, there is an explicit solution in terms of  classical Black-Scholes-Margrabe exchange option formulas, up to solving a scalar equation. We show this allows that prices may go negative under distortion of a positive-price lognormal model. An analogous construction applies when the lognormal base is switched to Bachelier normal dynamics. Illustrations based on WTI and nickel show the imbalance required to reproduce the extreme event prices under the benchmark assumptions.

\end{abstract}


\section{Introduction}
Commodity futures link financial trading to physical delivery. Storage, convenience yield and the term structure of supply and demand are central to their valuation; see, for example, \citet{Bouchouev2023virtualbarrels}. Participation by financial investors and the interaction between trading and fundamentals add further channels through which prices can move \citep{Tang2012index,Cheng2014}. During periods of market stress, however, the ability to hold a contract through delivery can differ sharply across participants. Traders who must close their positions may then face a limited set of counterparties able and willing to absorb them.

Two episodes motivate this paper. On April~20, 2020, the May WTI futures contract settled at $-\$37.63$/bbl on its penultimate trading day, amid a collapse in oil demand and severe constraints on storage and delivery \citep{CFTC2020WTI}. On March~8, 2022, the LME three-month nickel price reached $\$101{,}365$/ton during a short squeeze, before trading was suspended and that day's trades were canceled \citep{OliverWyman2023}. Both raise a related question: how can pressure to exit positions amplify a futures-price movement when access to physical delivery is limited?

We address this question through a model that separates traders by their license (or not) to take delivery. An imbalance among unlicensed positions creates one-sided demand to close-out: excess longs must sell, and excess shorts must buy. The cost of this constraint is represented by a roll option. Since the option payoff depends on the observed futures price, and the observed price includes the option-value correction, the pricing equation contains feedback, and is thereby nonlinear.

Our model is explanatory. It takes a benchmark price as an input, and builds an additional mechanism through which constrained positions can move the observed price away from that benchmark. Its main advantage is analytical tractability: despite the nonlinear pricing PDE, the specifications studied here admit exact solutions in terms of standard option-pricing formulas. These representations make the direction and magnitude of feedback transparent, and allow the implications of benchmark and imbalance assumptions to be quantified.

\subsection{Contract structure and motivating events}\label{sec:contracts}
The two episodes differ in their contract structure and delivery constraints, but both feature a class of market participants that cannot remain in the futures contract through physical settlement. We first describe the institutional features of the WTI and LME nickel markets and the circumstances surrounding the April 2020 oil's price collapse and the March 2022 nickel squeeze. We then show how these distinct settings give rise to the common mechanism studied in the paper: forced liquidation by delivery constrained traders and the resulting pressure on observed futures prices. Additional contract specifications and delivery details for the two markets are summarized in Appendix \ref{app:contracts}.

\subsubsection{NYMEX WTI crude oil and the April 2020 collapse}
The NYMEX West Texas Intermediate (WTI) crude oil futures contract is physically settled by delivering 1{,}000 barrels of WTI-grade crude at Cushing, Oklahoma. Trading ends shortly before the delivery month according to the schedule specified by the exchange \citep{CMERulebook200}. Physical settlement requires access to storage, pipeline capacity, and the operational arrangements necessary to receive or deliver crude at Cushing.

Many financial participants do not possess this infrastructure and therefore close or roll their front-month positions before the delivery process begins. Producers, refiners, merchants, and other commercial participants may have the arrangements needed to remain in the contract through settlement, although commercial status alone does not guarantee delivery capacity. Throughout the paper, we refer to these groups as unlicensed and licensed traders, respectively, where the terminology describes the delivery capability within the model rather than an actual regulatory classification.

The May~2020 WTI contract provides a striking example of a forced long liquidation, which we call Case~1. On April~20, 2020, its penultimate trading day, the contract initially traded at around $\$17$/bbl and ended the day at $-\$37.63$/bbl. A series of events over the previous weeks had driven down the price of physical barrels and amplified the downward pressure on the futures contract. The COVID-19 pandemic sharply reduced oil demand, while the breakdown of production negotiations between Russia and Saudi Arabia added to supply pressures \citep{CFTC2020WTI}. Although a replacement agreement was announced on April~12, it did not take effect until May~1.

Meanwhile, Cushing inventories reached 59.7 million barrels on April 17, up from 55.0 million barrels one week earlier \citep{EIAcushing}. Open interest in the expiring contract was 108{,}593 lots, corresponding to 108.6 million barrels \citep{CFTC2020WTI}. These conditions made physical delivery particularly unattractive and, for many financial participants, operationally infeasible. The concentration of the price collapse in the final hours is consistent with liquidation pressure, alongside the deterioration in physical-market conditions.

The open-interest data provide additional context, but do not identify which positions were delivery-constrained. Open interest fell from 108{,}593 to 13{,}044 contracts during April~20, implying that approximately $88\%$ of outstanding positions were closed during the day. Relative to its April~2 peak of 634{,}727 contracts, residual open interest on the penultimate trading day was only $2.06\%$, the smallest among the 198 WTI front-month expiration cycles from 2010 to 2024, compared with a median of $5.55\%$. We calculate these statistics using contract-level open-interest data obtained from Bloomberg. These figures document extensive position closure. They are consistent with a large price concession during liquidation, but aggregate open interest alone cannot establish the size or ownership of the constrained imbalance. This distinction motivates the mixing weight $\varepsilon$ introduced in Section~\ref{sec:market_structure}: it measures the intensity of the delivery-constrained positional imbalance relative to total open interest, rather than the fraction of positions remaining after liquidation.

A particularly important source of forced selling is documented by \citet[\S3.5]{Bouchouev2023virtualbarrels}. The Bank of China had offered an OTC oil-linked product, YuanYouBao, to domestic retail investors and hedged its exposure through long NYMEX WTI futures positions. Those positions had to be liquidated on April~20 through the Trading-at-Settlement (TAS) mechanism. As selling pressure intensified, TAS liquidity became insufficient even at the minimum permitted value, and the remaining orders had to be executed in the regular order book during the final part of the session \citep{Bouchouev2020negative}. In the terminology of our model, this is a Case~1 event, in which unlicensed traders are net long and must sell into a limited pool of licensed traders.

\subsubsection{LME three-month nickel and the March 2022 squeeze}

The March~2022 LME nickel squeeze provides the mirror-image Case~2 event. The LME nickel contract differs from WTI in several important respects. The LME offers daily prompt dates out to three months and quotes a rolling three-month tenor; an individual position nevertheless has a fixed delivery date. Physical delivery takes place through LME-approved warehouses and requires eligible Class~I refined nickel. Much of global nickel production, however, consists of Class~II material that is not deliverable against the LME contract. Consequently, even a physical nickel producer may be unable to satisfy an LME short position using its own production \citep{OliverWyman2023}.

This distinction became critical during March~7--8, 2022. Concerns about the availability of deliverable nickel intensified
following Russia's invasion of Ukraine, contributing to reduced
liquidity ahead of the squeeze \citep{OliverWyman2023}. Warranted nickel stocks in LME warehouses had declined from 173,000 tonnes on September 15, 2021, to approximately 80,000 tonnes by February 24, 2022
\citep[p.~15]{OliverWyman2023}. Against this backdrop, the nickel futures price rose sharply, reaching $\$101{,}365$/ton on March~8 before the LME suspended trading and subsequently canceled that day's transactions. The episode led to an FCA enforcement investigation \citep{FCA2023LME} and litigation concerning the cancellation of trades \citep{Elliott2023}.

Large short positions and rising margin requirements contributed
to the squeeze. Some producers used LME contracts to hedge Class II nickel, which was not deliverable against those contracts. As prices
rose, participants reduced their short exposures through purchases, adding further upward pressure in a market with limited liquidity \citep{OliverWyman2023}. This mechanism illustrates Case~2:
delivery constrained shorts face pressure to buy back their positions.

\subsubsection{A common feedback mechanism}

The episodes motivate two directions of the same feedback mechanism. In Case~1, constrained traders are net long and closing their positions puts downward pressure on the observed price. In Case~2, they are net short and covering puts upward pressure on the price. For WTI, the relevant constraint is access to storage and delivery arrangements at Cushing. For nickel, it includes the mismatch between a producer's output and the exchange's deliverable grade.
Section~\ref{sec:market_structure} formalizes the position imbalance.

Our approach is related to models in which trading changes the price dynamics faced by the traders themselves. \citet{FreyStremme1997} examine how the implementation of dynamic hedging strategies changes the underlying asset's price dynamics and volatility. \citet{SircarPapanicolaou1998} account for the effect of dynamic hedging on the underlying asset price and, in turn, on derivative values. Their construction leads to a nonlinear pricing PDE and explains how hedging demand can increase volatility. In commodity markets, \citet*{Sircar2015feedback} introduce portfolio-optimizing investors alongside traditional commodity demand. Their feedback model studies how investors' positions affect commodity prices and their correlation with financial markets. \citet{pavlova2016model} develop an equilibrium model with institutional investors whose performance is measured relative to a commodity index. They show how financialization affects futures prices, volatilities and correlations, and transmits shocks to spot prices and inventories when storage is included. \citet{SockinXiong2015} examine an informational feedback channel in which commodity prices influence producers' decisions by conveying information about economic conditions, thereby affecting commodity demand. More directly related to the WTI episode, \citet{Ma2021strategic} develop a three-player game-theoretic model in which strategic interactions can generate negative oil prices.

The common feature is that financial trading enters price formation. Our model focuses on a specific source of this interaction: the need to close positions when physical delivery is unavailable. Feedback acts through a roll-option value, which depends on the same observed price that it helps determine. The benchmark is supplied externally, and the imbalance is held fixed within each valuation. The model therefore complements the broader equilibrium and trading mechanisms above. Its main analytical contribution is an exact reduction of this particular nonlinear pricing problem to standard option formulas, coupled with a scalar equation.

\subsection{Contributions}
The paper makes three related contributions. First, it links market segmentation to a pricing distortion. A normalized imbalance of unlicensed long and short positions determines the weight $\varepsilon$, while a roll option measures the cost associated with closing the excess position. The resulting feedback relation accommodates both downward and upward distortions. 

Second, we derive and solve the associated nonlinear pricing problem. Under lognormal benchmark dynamics, homogeneity reduces the PDE to one state variable, and a change of variables converts it to the classical Margrabe equation. The resulting exact representation expresses the solution through a standard put or call Black-Scholes pricing formula. 
Under normal dynamics, an analogous reduction gives Bachelier representations. 

Third, we use the representation to characterize the implications of the feedback mechanism and to construct event illustrations. In particular, we show how the distortion can drive an observed oil futures price below zero even when the underlying licensed price remains positive, and we identify the role of the position imbalance and calendar-spread dynamics in determining when such a crossing can occur. We then use the WTI and nickel episodes to illustrate the magnitude and direction of the resulting price distortions.

\section{Feedback Model}

As discussed earlier, commodity futures markets differ fundamentally from equity markets because they include distinct classes of participants with different objectives and delivery capabilities. Producers, merchants, and commodity trading firms primarily use futures contracts to hedge the price risk of physical inventories or future production. By contrast, financial participants trade futures for speculation, portfolio diversification, or financial hedging.

To capture this segmentation, we distinguish between two types of traders, based on the delivery rules of major commodity exchanges such as the CME.

\begin{enumerate}
    \item \textbf{Licensed traders}. These are producers, merchants, and commodity trading firms that are authorized to make or take physical delivery at contract maturity.

    \item \textbf{Unlicensed traders}. These include hedge funds and other financial participants who trade commodity futures for speculative or hedging purposes but are not authorized to participate in physical delivery.
\end{enumerate}

The key distinction is that an unlicensed trader must close its position before contract maturity, whereas a licensed trader may hold the contract until expiration and either make or take physical delivery. This reflects the market structure observed in both canonical events, where speculative participants were required to liquidate their positions before maturity. Maintaining an open position beyond that point without delivery eligibility is generally prohibited by exchange rules and may result in financial penalties or trading restrictions.

We introduce this market segmentation to describe the observed futures price, defined as $F^{\mathrm{obs}}_{t,T}$. Let $F^{\mathrm{lic}}_{t,T}$ denote the licensed price, namely the price that would prevail if all market participants were licensed and no forced liquidation occurred. The difference between the two arises from the presence of the unlicensed traders and is captured by a third quantity, the distorted futures price $H_{t,T,T_2}$. This is the price at which the unlicensed component of the market clears.

Throughout the paper, $F_{t,T_2}$ denotes the observed price of a later-month futures contract, where $T_2>T$. Since this contract matures well after the front-month contract, it is not subject to the same liquidation pressure and is therefore treated as undistorted. In the WTI oil application, this is the next-month contract.

The following sections derive the relationship between $F^{\mathrm{obs}}_{t,T}$, $F^{\mathrm{lic}}_{t,T}$, $H_{t,T,T_2}$, and $F_{t,T_2}$ from a market-clearing argument based on the open-interest composition of the futures market.

\subsection{Market Structure and Open Interest}\label{sec:market_structure}

At any time \(t < T\), the open interest in the contract maturing at \(T\) decomposes by trader type and side of the trade. Let \(\ell^L_t, \ell^U_t\) denote the long positions held by licensed and unlicensed traders at time \(t\), respectively, and \(s^L_t, s^U_t\) the corresponding short positions. We summarize the decomposition in the open-interest matrix
\begin{equation}\label{eq:oi_matrix}
\begin{array}{c|cc}
 & \text{Long} & \text{Short} \\
\hline
\text{Licensed} & \ell^L_t & s^L_t \\
\text{Unlicensed} & \ell^U_t & s^U_t
\end{array}
\end{equation}
Since every futures contract has exactly one long and one short counterparty, the total open interest is 
\begin{equation}\label{eq:oi_balance}
    \mathrm{OI}_t \;=\; \ell^L_t + \ell^U_t \;=\; s^L_t + s^U_t.
\end{equation}

The delivery constraint, as we discussed earlier, introduces an asymmetry in the two rows of the table. In particular, by maturity \(T\), all unlicensed positions must be closed, while licensed positions can remain open and proceed to physical settlement. 
The delivery constraint introduces an asymmetry between the two classes of traders. Let \(T'<T\) denote the time by which unlicensed positions must be closed before physical settlement at \(T\). Licensed positions may remain open beyond \(T'\) and proceed to delivery. Consequently, the net unlicensed exposure immediately prior to liquidation,
\[
\ell_{T'}^U-s_{T'}^U,
\]
determines the direction of the forced order flow. This creates two distinct regimes identified as Case 1 and Case 2 for the rest of our paper.

\subsubsection*{Case 1: Net unlicensed long positions ($\ell^U_{T'} > s^U_{T'}$)}
The first case occurs when within the unlicensed segment the long positions outnumber the short positions. After the $\ell_{T'}^U-s_{T'}^U$ excess unlicensed long positions are closed against licensed shorts, the number of surviving licensed shorts is
$s_{T'}^L-\left(\ell_{T'}^U-s_{T'}^U\right)=\ell_{T'}^L$,
where the equality follows from the open-interest identity \eqref{eq:oi_balance}. Thus, after all unlicensed positions have been cleared, the remaining licensed market is balanced, with $\ell_{T'}^L$ long and $\ell_{T'}^L$ short positions, consistent with physical delivery.

We model the observed futures price as the volume-weighted average of 
quantities associated with the two classes of traders: 
\begin{equation}\label{eq:f_obs_case1}
    F^{\text{obs}}_{t,T} \;=\; \frac{\ell^L_{T'} + s^U_{T'}}{\mathrm{OI}_{T'}}\,F^{\text{lic}}_{t,T} \;+\; \frac{\ell^U_{T'} - s^U_{T'}}{\mathrm{OI}_{T'}}\,H_{t,T,T_2}.
\end{equation}
The weights are nonnegative and sum to one by \eqref{eq:oi_balance}. The deviation of \(F^{\text{obs}}_{t,T}\) from the fundamental is thus governed by the gap between \(H_{t,T,T_2}\) and \(F^{\text{lic}}_{t,T}\), scaled by the net unlicensed long imbalance \((\ell^U_{T'} - s^U_{T'})/\mathrm{OI}_{T'}\). The explicit characterization of \(H_{t,T,T_2}\) is deferred to Section~\ref{sec:speculative}.

\subsubsection*{Case 2: Net unlicensed short positions ($s^U_{T'} > \ell^U_{T'}$)}
The mirror situation occurs when unlicensed shorts dominate the segment. A net unlicensed short of size \(s^U_{T'} - \ell^U_{T'}\) must close by buying from licensed longs. By the analogous argument, after all unlicensed positions clear, the remaining licensed positions satisfy the volume-weighted average observed futures price given by
\begin{equation}\label{eq:f_obs_case2}
    F^{\text{obs}}_{t,T} \;=\; \frac{s^L_{T'} + \ell^U_{T'}}{\mathrm{OI}_{T'}}\,F^{\text{lic}}_{t,T} \;+\; \frac{s^U_{T'} - \ell^U_{T'}}{\mathrm{OI}_{T'}}\,H_{t,T,T_2},
\end{equation}
where here \(H_{t,T,T_2}\) denotes the distorted price of an unlicensed short, characterized in Section~\ref{sec:speculative} via Corollary~\ref{cor:short_side_pde}.

Equations \eqref{eq:f_obs_case1} and \eqref{eq:f_obs_case2} share the structure of a convex combination between the licensed price $F^{\text{lic}}_{t,T}$ and the distorted price $H_{t,T,T_2}$. Defining the net unlicensed imbalance, as a fraction of the net unlicensed exposure over total open interest
\begin{equation}\label{eq:g_endogenous}
\varepsilon \;:=\; \frac{|\ell^U_{T'} - s^U_{T'}|}{\mathrm{OI}_{T'}},
\end{equation}
we obtain the unified pricing formula
\begin{equation}\label{eq:f_obs_formula}
    F^{\text{obs}}_{t,T} \;=\; (1 - \varepsilon)\,F^{\text{lic}}_{t,T} \;+\; \varepsilon\,H_{t,T,T_2}.
\end{equation}
\begin{assumption}
    Throughout this paper, we assume $\varepsilon<1$.
\end{assumption}

\subsection{Fundamental Futures Pricing}\label{sec:fundamental}
The framework we develop in this paper is built on top of any standard commodity-pricing model for the licensed futures price \(F^{\mathrm{lic}}_{t,T}\) in equation~\eqref{eq:f_obs_formula}.
The rest of the equation provides a structural mechanism that maps this fundamental price into the observed futures price.

Several standard frameworks can supply \(F^{\mathrm{lic}}\) to our model. These include classical cost-of-carry models with storage costs and convenience yields, lognormal dynamics commonly used for futures prices, and the models of \citet{Schwartz1997}, which incorporate
mean reversion in commodity prices and, in their multi-factor specifications, stochastic convenience yields and interest rates. These models, together with a broader overview of the commodity-pricing literature, are presented in \citet{Bouchouev2023virtualbarrels}.

These classical frameworks produce a fundamental futures price that is, by construction, positive and consistent with the absence of arbitrage. The mechanism introduced in the following subsection adds a structural correction that generates the extreme price distortions observed in practice due to exit pressure. The main structural component of our framework is therefore the mapping \(F^{\mathrm{lic}} \mapsto F^{\mathrm{obs}}\) through the recursive feedback equation derived below in Section~\ref{sec:speculative}.

\subsection{Speculative Futures Pricing}\label{sec:speculative}
The distorted futures price follows from a market-clearing argument. We consider separately the two possible imbalances in unlicensed positions.

In Case~1, unlicensed traders hold an excess long position that must be closed before maturity. These traders are therefore forced to sell their contracts, while only licensed traders can take the opposite side of the transaction. This creates one-sided selling pressure and can push the observed futures price below its fundamental value. 
The cost of this restriction can be represented through a roll option. A licensed trader can hold the front-month futures contract through maturity and take delivery of the commodity. An unlicensed trader cannot. To maintain the same exposure beyond maturity, the unlicensed trader must instead close the front-month position and replace it with a position in the next-month futures contract. We model this transaction as the long-side roll option, \(R^{(+)}\). 
The value of this option measures the cost associated with being unable to take delivery and being forced to close the front-month position. We therefore interpret \(R^{(+)}\) as the close-out premium borne by an unlicensed long trader. The corresponding distorted price is $H_{t,T,T_2}=F^{\text{lic}}_{t,T}-R^{(+)}_t$.

In Case~2, the direction of the imbalance is reversed. Unlicensed traders hold an excess short position that must be closed before maturity. They are therefore forced to buy the front-month contract to cover their positions. This one-sided buying pressure can push the observed futures price above its fundamental value. 
The corresponding restriction is represented by the short-side roll option \(R^{(-)}\). Its value measures the additional cost faced by an unlicensed short trader who must buy back the front-month contract rather than settle through delivery. We interpret \(R^{(-)}\) as the squeeze premium generated by forced short covering: $H_{t,T,T_2}=F^{\text{lic}}_{t,T}+R^{(-)}_t$.

We now define these two roll options more precisely. Both are non-negative option payoffs written on the spread between the front-month and a later-month contract. Which option is relevant is determined by the open-interest imbalance: \(R^{(+)}\) applies when unlicensed traders are net long, while \(R^{(-)}\) applies when they are net short.

Because the underlying commodity is not itself directly tradable on the futures exchange, we represent its value at time \(T\) using the discounted value of the next-month futures contract. Let \(F^{\text{obs}}_{t,T}\) denote the observed futures price at time $t$ of the front-month futures contract whose maturity is \(T\), and let \(F^{\text{obs}}_{T,T_2}\) denote the futures price at time \(T\) of the contract maturing at \(T_2>T\). Let
\(\Delta = T_2-T\) denote the time between the two maturities, typically one month, and let \(D := e^{-r\Delta}\) denote the discount factor from \(T_2\) to \(T\), assumed constant. The factor $D$ discounts between the two delivery dates $T$ and
$T_2$. Separately, we express the roll option terms in futures price units: they represent conditional expected terminal payoffs rather than time-$t$ present values. Under deterministic interest rates, their discounted counterparts are conventional option values, as in \citet{black1976futures}. 

The terminal payoffs of the two roll options are then
\begin{equation}
    R^{(+)}_T
    =
    \left(
        DF^{\text{obs}}_{T,T_2}
        -
        F^{\text{obs}}_{T,T}
    \right)^{+},
    \qquad
    R^{(-)}_T
    =
    \left(
        F^{\text{obs}}_{T,T}-DF^{\text{obs}}_{T,T_2}
    \right)^{+}.
\label{eq:r_terminal_val}
\end{equation}
Thus, \(R^{(+)}_T\) is positive when the front-month contract trades below the discounted next-month contract (contango), while \(R^{(-)}_T\) is positive when the front-month contract trades above it (backwardation). The two options therefore capture distortions in opposite directions.

The next-month futures contract is itself traded on the exchange. However, because its maturity \(T_2\) remains sufficiently far from the front-month maturity \(T\), it is not yet exposed to the same forced-liquidation mechanism. We therefore assume that its observed price coincides with its licensed price over the relevant period.
An unlicensed trader can trade only the market futures contract and must close any remaining position by the end of the last trading session. Consequently, the trader is exposed to the prevailing observed price \(F^{\text{obs}}_T\) at forced liquidation.

Combining this with \eqref{eq:f_obs_formula}, we obtain
\begin{equation}\label{eq:f_obs_reduced}
    F^{\text{obs}}_{t,T}
    =
    \begin{cases}
        F^{\text{lic}}_{t,T}
        -
        \varepsilon R^{(+)}_{t},
        &
        \ell^{U}_{T'} > s^{U}_{T'}
        \quad \text{(Case 1)},\\[4pt]
        F^{\text{lic}}_{t,T}
        +
        \varepsilon R^{(-)}_{t},
        &
        \ell^{U}_{T'} < s^{U}_{T'}
        \quad \text{(Case 2)},
    \end{cases}
    \qquad
    \varepsilon
    =
    \frac{
        \left|\ell^{U}_{T'}-s^{U}_{T'}\right|
    }{
        \mathrm{OI}_{T'}
    }.
\end{equation}
The deviation from the licensed price comes from the roll option, which determines the value of the constraint faced by unlicensed traders, while \(\varepsilon\) measures the size of the unlicensed position imbalance relative to total open interest. 

In Case~1, \(R^{(+)}_t\) measures the close-out premium associated with forced selling by unlicensed long traders. Because these traders cannot take delivery, they must liquidate their positions even when the front-month price is depressed. This mechanism pushes
$F^{\text{obs}}_{t,T}<F^{\text{lic}}_{t,T}$, 
as in the April 2020 WTI episode. 
In Case~2, \(R^{(-)}_t\) measures the squeeze premium associated with forced buying by unlicensed short traders. These traders must purchase contracts to close their short positions, generating upward price pressure and potentially producing $F^{\text{obs}}_{t,T}>
F^{\text{lic}}_{t,T}$, 
as in a classical short squeeze, such as the March 2022 nickel episode.

Thus, the same market-clearing mechanism generates distortions in either direction. An excess unlicensed long position activates \(R^{(+)}\) and creates downward pressure, while an excess unlicensed short position activates \(R^{(-)}\) and creates upward pressure. The composition of open interest determines both the direction of the distortion and, through \(\varepsilon\), its magnitude.

Since the roll option term adjusts the futures quote, we use
its undiscounted conditional expectation:
\begin{equation}\label{eq:forward_convention}
    R^{(\pm)}_t
    :=
    \mathbb{E}^{\mathbb{Q}}
    \!\left[
        R^{(\pm)}_T
        \mid
        \mathcal{F}_t
    \right],
    \qquad t\in[0,T].
\end{equation}
Consequently, the $R^{(\pm)}$ terms reference $F^{\text{obs}}$ through \eqref{eq:r_terminal_val}, which then {\em feeds back} into $F^{\text{obs}}$ through \eqref{eq:f_obs_reduced}. Putting it together gives
\begin{equation}\label{eq:feedback}
    F^{\text{obs}}_{t,T}
    =
    \begin{cases}
        F^{\text{lic}}_{t,T}-\varepsilon \mathbb{E}^{\mathbb{Q}}\left[\left(DF^{\text{obs}}_{T,T_2} - F^{\text{obs}}_{T,T}\right)^{+}\mid  \mathcal{F}_t\right],
        &\text{(Case 1)},\\
        \\
        F^{\text{lic}}_{t,T}
        +        \varepsilon \mathbb{E}^{\mathbb{Q}}
    \left[
      \left(
        F^{\text{obs}}_{T,T}-DF^{\text{obs}}_{T,T_2}
    \right)^{+}
        \mid
        \mathcal{F}_t
    \right],
        &\text{(Case 2)},
        \end{cases}
\end{equation}
which reveals the nonlinearity of the pricing mechanism. 

\subsection{Roll-Option Pricing PDE}\label{sec:pde}

Pricing the distortion thus reduces to pricing the operative roll-option value \(V^\pm\), with one essential complication. The terminal payoff~\eqref{eq:r_terminal_val} of \(R^{(\pm)}\) involves \(F^{\text{obs}}_T\), which is itself constructed from \(V^\pm\) through \eqref{eq:f_obs_reduced}. The option is therefore self-referential and coupled to its own price through the feedback equation, and the standard Black-style replication argument must be extended to internalize this coupling. 

We work with the following geometric Brownian motion (GBM) model. Under the risk--neutral measure \(\mathbb{Q}\), suppose
\begin{equation}
\frac{dF^{\text{lic}}_{t,T}}{F^{\text{lic}}_{t,T}} = \sigma_1(t)\,dW^1_t,
\qquad
\frac{dF^{\text{obs}}_{t,T_2}}{F^{\text{obs}}_{t,T_2}} = \sigma_2(t)\,dW^2_t, \label{dF2eqn}
\qquad
d\langle W^1, W^2\rangle_t = \rho\,dt,
\end{equation}
 where \(W^1,W^2\) are \(\mathbb{Q}\)-Brownian motions, \(\sigma_1(\cdot),\sigma_2(\cdot)>0\), and the correlation coefficient \(\rho\in(-1,1)\) is constant.
We will see that, as our feedback mechanism is implemented pointwise in time, the resulting dynamics for $(F^{\text{obs}}_{t,T}, F^{\text{obs}}_{t,T_2})$ will be Markovian with respect to the filtration generated by \((W^1,W^2)\), and so we look for 
 \[ R_t^{\pm} = V^{\pm}(t, F^{\text{obs}}_{t,T}, F^{\text{obs}}_{t,T_2}), \]
 where \(V^{\pm}(t,f_1,f_2)\) are deterministic functions:
 \begin{equation}
  V^{\pm}(t,f_1,f_2) = \mathbb{E}^{\mathbb{Q}}\!\bigl[R^{(\pm)}_T \,\big|\, F^{\text{obs}}_{t,T}=f_1, F^{\text{obs}}_{t,T_2}=f_2\bigr], \quad t \in [0, T], f_1\in\R_+, f_2\in\R_+,
  \label{Vpmdef}
\end{equation}
with $R^{\pm}_T$ defined in \eqref{eq:r_terminal_val}. 
The possible values $f_2$ of $F^{\text{obs}}_{t,T_2}$ are strictly positive, the next month futures price is a GBM (and it is not affected by the distortion feedback). We also take $f_1$ to be strictly positive for now, since traders are pricing under the presumption that $F^{\text{obs}}_{t,T}$ also cannot go negative since the baseline futures price $F^{\text{lic}}_{t,T}$ is also GBM. We address how the model changes if the observed $F^{\text{obs}}_{t,T}$ does become negative in Section \ref{sec:bachelier}.

We start with Case 1, where
\begin{equation}\label{eq:long_feedback_relation}
    F^{\text{obs}}_{t,T}
    \;=\;
    F^{\text{lic}}_{t,T}
    - \varepsilon V^+\!\left(t,F^{\text{obs}}_{t,T},F^{\text{obs}}_{t,T_2}\right).
\end{equation}
The following proposition characterizes \(V^{+}\) as the solution of a nonlinear feedback PDE. 
\begin{proposition}\label{prop:roll_pde}
Suppose that $V^+(t, f_1, f_2)\in\mathcal{C}^{1,2,2}([0,T)\times\R_+\times\R_+)$. Then $V^+$ solves
\begin{equation}\label{eq:pde_vplus}
    V_t^+
    + \tfrac{1}{2}\,J^+\,\sigma_1^2(t)\,f_1^2\,V^+_{f_1f_1}
    + K^+\,\sigma_1(t)\sigma_2(t)\,f_1 f_2\,V^+_{f_1f_2}
    + \tfrac{1}{2}\,\sigma_2^2(t)\,f_2^2\,V^+_{f_2f_2}
    = 0,
\end{equation}
with diffusion multipliers
\begin{equation*}\label{eq:JK_plus}
    J^+ := \frac{(B^+)^2 + \varepsilon^2 (C^+)^2 - 2\rho\varepsilon B^+ C^+}{(A^+)^2},
    \qquad
    K^+ := \frac{\rho B^+ - \varepsilon C^+}{A^+},
\end{equation*}
given in terms of the feedback coefficients
\begin{equation}\label{eq:ABC_plus}
    A^+ := 1 + \varepsilon V^+_{f_1},
    \qquad
    B^+ := 1 + \varepsilon\,\frac{V^+}{f_1},
    \qquad
    C^+ := \left(\frac{f_2}{f_1}\right) \frac{\sigma_2(t)}{\sigma_1(t)}\,V^+_{f_2}.
\end{equation}
The terminal condition is
\begin{equation}
    V^+(T,f_1, f_2) \;=\; \bigl(D f_2 - f_1\bigr)^{+}.
    \label{terminalplus}
\end{equation}
\end{proposition}
\begin{proof}
Under the forward-pricing convention~\eqref{eq:forward_convention}, $V^+$ and $F^{\text{lic}}_{t,T}$ are $\mathbb{Q}$-martingales.
As the difference of two martingales, $F^{\text{obs}}_{t,T}$ is an It\^o martingale, and it admits the representation
\begin{equation}\label{eq:fobs-dyn}
    dF^{\text{obs}}_{t,T} \;=\; \beta_t\,dW^1_t + \gamma_t\,dW^2_t.
\end{equation}
Applying It\^o's formula to $V^+\!\left(t,F^{\text{obs}}_{t,T},F^{\text{obs}}_{t,T_2}\right)$ and using \eqref{eq:fobs-dyn} together with \eqref{dF2eqn}, we have 
\begin{align}\label{eq:ito-V}
    dV^+ \;=\; & V^+_t\,dt + V^+_{f_1}\bigl(\beta_t\,dW^1_t + \gamma_t\,dW^2_t\bigr) + V^+_{f_2}\,f_2\sigma_2(t)\,dW^2_t\notag\\
    &+ \tfrac{1}{2}V^+_{f_1 f_1}\bigl(\beta_t^2 + 2\rho\beta_t\gamma_t + \gamma_t^2\bigr)\,dt + \tfrac{1}{2}V^+_{f_2 f_2}\,f_2^2\sigma_2^2(t)\,dt + V^+_{f_1 f_2}\bigl(\rho\beta_t + \gamma_t\bigr)\,f_2\sigma_2(t)\,dt.
\end{align}
On the other hand, from \eqref{eq:long_feedback_relation}, we have
\begin{equation}\label{eq:rep-V}
    \varepsilon\,dV^+ \;=\; \Bigl[\bigl((f_1 + \varepsilon V^+)\sigma_1(t) - \beta_t\bigr)\,dW^1_t \;-\; \gamma_t\,dW^2_t\Bigr].
\end{equation}
Equating the $dW^1_t$ coefficients in \eqref{eq:ito-V} and \eqref{eq:rep-V} yields
\begin{equation}\label{eq:beta}
    \beta_t \;=\; \frac{(f_1 + \varepsilon V^+)\,\sigma_1(t)}{A^+}, \qquad A^+ := 1 + \varepsilon V^+_{f_1},
\end{equation}
while from the $dW^2_t$ coefficients, we obtain
\begin{equation}\label{eq:gamma}
    \gamma_t \;=\; -\,\frac{\varepsilon V^+_{f_2}\,f_2\sigma_2(t)}{A^+}.
\end{equation}
Finally, from the $dt$ coefficients, as the right-hand side of \eqref{eq:rep-V} is purely stochastic, we have
\begin{equation}\label{eq:dV-drift}
    0 \;=\; V^+_t + \tfrac{1}{2}V^+_{f_1 f_1}\bigl(\beta_t^2 + 2\rho\beta_t\gamma_t + \gamma_t^2\bigr) + \tfrac{1}{2}V^+_{f_2 f_2}\,f_2^2\sigma_2^2(t) + V^+_{f_1 f_2}\bigl(\rho\beta_t + \gamma_t\bigr)\,f_2\sigma_2(t).
\end{equation}
Substituting \eqref{eq:beta} and \eqref{eq:gamma} into \eqref{eq:dV-drift} yields the PDE \eqref{eq:pde_vplus}. The terminal condition \eqref{terminalplus} follows directly from the payoff in \eqref{eq:r_terminal_val}.
\end{proof}

\begin{corollary}[Case 2 roll-option PDE]\label{cor:short_side_pde}
Suppose that \(V^-(t,f_1,f_2)\in \mathcal{C}^{1,2,2}([0,T)\times\R_+\times\R_+)\).
In Case~2, unlicensed short positions must be covered by buying before maturity, so
\[
F^{\mathrm{obs}}_{t,T}
=
F^{\mathrm{lic}}_{t,T}
+
\varepsilon
V^-\!\left(t,F^{\mathrm{obs}}_{t,T},F^{\mathrm{obs}}_{t,T_2}\right).
\]
Then \(V^-\) solves
\begin{equation}\label{eq:pde_vminus}
    V^-_t
    + \tfrac{1}{2}\,J^-\,\sigma_1^2(t)\,f_1^2\,V^-_{f_1f_1}
    + K^-\,\sigma_1(t)\sigma_2(t)\,f_1 f_2\,V^-_{f_1f_2}
    + \tfrac{1}{2}\,\sigma_2^2(t)\,f_2^2\,V^-_{f_2f_2}
    = 0,
\end{equation}
with terminal condition
\begin{equation}\label{eq:vminus_terminal}
    V^-(T,f_1,f_2)
    =
    \bigl(f_1-Df_2\bigr)^+,
\end{equation}
and diffusion multipliers
\begin{equation}\label{eq:JK_minus}
    J^- :=
    \frac{(B^-)^2+\varepsilon^2(C^-)^2
    +2\rho\varepsilon B^-C^-}{(A^-)^2},
    \qquad
    K^- :=
    \frac{\rho B^-+\varepsilon C^-}{A^-}
\end{equation}
given in terms of the feedback coefficients
\begin{equation}\label{eq:ABC_minus}
    A^- := 1-\varepsilon V^-_{f_1},
    \qquad
    B^- := 1-\varepsilon\frac{V^-}{f_1},
    \qquad
    C^- :=
    \left(\frac{f_2}{f_1}\right)
    \frac{\sigma_2(t)}{\sigma_1(t)}
    V^-_{f_2}.
\end{equation}
\end{corollary}

\begin{proof}
The Case~2 feedback relation can be written as
\[
F^{\mathrm{obs}}_{t,T}
=
F^{\mathrm{lic}}_{t,T}
-
(-\varepsilon)
V^-\!\left(t,F^{\mathrm{obs}}_{t,T},F^{\mathrm{obs}}_{t,T_2}\right).
\]
Hence the derivation of Proposition~\ref{prop:roll_pde} applies with
\(V^+\) replaced by \(V^-\) and \(\varepsilon\) replaced by
\(-\varepsilon\). This changes the sign of every term containing a single
factor of \(\varepsilon\), yielding \eqref{eq:ABC_minus} and
\eqref{eq:JK_minus}. The terminal condition \eqref{eq:vminus_terminal} follows from the Case~2 payoff
in \eqref{eq:r_terminal_val}.
\end{proof}
We note that the feedback PDEs \eqref{eq:pde_vplus} and \eqref{eq:pde_vminus} are quasilinear, in that the nonlinearities from the diffusion multiplier come as functions of $(V^\pm,V^\pm_{f_1},V^\pm_{f_2})$, but not of the highest second-order derivatives $V^\pm_{f_if_j}$. 

\subsection{Margrabe dimension reduction and explicit formulas}
When there is no distortion, that is $\varepsilon=0$, we have
\[ A^\pm=B^\pm=J^\pm=1, \qquad K^\pm=\rho, \]
and the $V^\pm$ reduce to the standard calendar exchange option pricing functions  under a bivariate lognormal model: $V^\pm\equiv V^{(0)\pm}$, where $V^{(0)\pm}(t,f_1,f_2)$ solve the linear PDE  and terminal conditions 
\begin{align}
      &V^{(0)\pm}_t
    + \tfrac{1}{2}\,\sigma_1^2(t)\,f_1^2\,V^{(0)\pm}_{f_1f_1}
    + \rho\sigma_1(t)\sigma_2(t)\,f_1 f_2\,V^{(0)\pm}_{f_1f_2}
    + \tfrac{1}{2}\,\sigma_2^2(t)\,f_2^2\,V^{(0)\pm}_{f_2f_2}
    = 0,\label{margrabePDE}\\
    & V^{(0)+}(T,f_1,f_2) = (Df_2 - f_1)^+, \qquad V^{(0)-}(T,f_1,f_2) = (f_1-Df_2)^+. \nonumber
\end{align}

From this, one can show that $V^{(0)\pm}$ are homogeneous of degree one in \((f_1,f_2)\), that is 
$$V^{(0)\pm}(t,\lambda f_1,\lambda f_2)\equiv \lambda V^{(0)\pm}(t,f_1,f_2), \quad \mbox{for all } \lambda>0.$$
Consequently, we have the following explicit formulas. 
\begin{lemma}\label{Marg_lemma}
The exchange option pricing functions $V^{(0)\pm}$, solutions of the PDE problems \eqref{margrabePDE}, admit the representation
\begin{equation*}
    V^{(0)+}(t,f_1,f_2)=f_2P(t,\xi),\quad  V^{(0)-}(t,f_1,f_2)=f_2C(t,\xi),
\qquad
\xi:=f_1/f_2, 
\end{equation*}
where the put price $P$ solves the one-dimensional linear PDE 
\begin{equation}
P_t + \frac12\sigma_M^2(t)\xi^2P_{\xi\xi}
=0,\qquad t<T, \, \xi>0, \label{MargPDE}
\end{equation}
with terminal condition $P(T,\xi)=\bigl(D-\xi\bigr)^+$, and where
\begin{equation}
    \label{eq:sigmaM_def}
\sigma_M^2(t)=\sigma_1^2(t)+\sigma_2^2(t)-2\rho\,\sigma_1(t)\sigma_2(t).
\end{equation} 
The call option price function $C$ solves the same PDE \eqref{MargPDE} as $P$, but with terminal condition $C(T,\xi)=\bigl(\xi-D\bigr)^+$. 
The Black-Scholes-Margrabe put price function is given explicitly:
\begin{equation}
  P(t,\xi) = D\,N\!\left(-d_2(t,\xi)\right) - \xi\,N\!\left(-d_1(t,\xi)\right), \label{Pformula}
\end{equation}
while the corresponding call price function is 
\begin{equation}    
  C(t,\xi) = \xi\,N\!\left(d_1(t,\xi)\right) - D\,N\!\left(d_2(t,\xi)\right),  \label{Cformula}
\end{equation}
where 
\begin{equation}
d_{1,2}(t,\xi)
=
\frac{\log(\xi/D)\pm\frac12\Sigma_M^2(t,T)
}{
\Sigma_M(t,T)
},
\qquad
\Sigma_M(t,T):= \sqrt{\int_t^T \sigma_M^2(s)\,ds},
\label{eq:d1d2}
\end{equation}
and $N$ is the standard normal c.d.f. 
\end{lemma}
\begin{proof}
The formulas follow from the exchange option formula of \citet{Margrabe1978}, with deterministic time-dependent volatilities incorporated through the integrated variance
$\Sigma_M^2(t,T)$.
\end{proof}

\subsection{Dimension reduction for nonlinear feedback PDEs}\label{sec:reduction}
We observe that the diffusion multipliers $J^+$ and $K^+$ in the PDE \eqref{eq:pde_vplus} are functions of $B^+/A^+$ and $C^+/A^+$. Under the scaling
\[ V^+\mapsto\lambda^{-1}V^+, \quad f_1\mapsto\lambda^{-1}f_1, \quad  f_2\mapsto\lambda^{-1}f_2,  \]
we have 
\[ \frac{B^+}{A^+} = \frac{f_1+\varepsilon V^+}{f_1+\varepsilon f_1V^+_{f_1}}, \qquad
\frac{C^+}{A^+} = \frac{f_2\sigma_2(t)V^+_{f_2}}{\sigma_1(t)(f_1+\varepsilon f_1V^+_{f_1})},
\]
both independent of $\lambda>0$. 

By analogy with the Margrabe dimension-reduction for the case $\varepsilon=0$ in Lemma \ref{Marg_lemma}, we look for a solution for $V^+$ of the form 
$V^+(t,f_1,f_2)=f_2G^+(t,\xi)$, 
where $G^+$ depends on $(f_1,f_2)$ through the calendar ratio $\xi:=f_1/f_2$. Remarkably, the dimension reduction is possible even for the nonlinear PDE with feedback $\varepsilon>0$.

\begin{proposition}[Case 1 dimension reduction]
\label{prop:long_dimension_reduction}
Let \(V^+(t,f_1,f_2)\) be a classical solution of the Case~1 PDE
\eqref{eq:pde_vplus}, with terminal condition \eqref{terminalplus}.
Then we have
\begin{equation}
    V^+(t,f_1,f_2)=f_2G^+(t,\xi),
\qquad
\xi:=f_1/f_2, 
\label{Vplusreduc}
\end{equation}

where the reduced function \(G^+\) solves the one-dimensional nonlinear PDE
\begin{equation}
\label{eq:Gplus_reduced}
G^+_t
+
\frac12
\left(\frac{B^+}{A^+}\right)^2
\sigma_M^2(t)\xi^2G^+_{\xi\xi}
=0,\qquad t<T, \xi>0, 
\end{equation}
with terminal condition
\begin{equation}
\label{eq:Gplus_terminal}
G^+(T,\xi)=\bigl(D-\xi\bigr)^+,
\end{equation}
and where
\begin{equation}
\label{eq:ABplus_reduced}
A^+:=1+\varepsilon G^+_\xi,
\qquad
B^+:=1+\varepsilon\frac{G^+}{\xi}.
\end{equation}
\end{proposition}

\begin{proof}
The result follows from inserting the form \eqref{Vplusreduc} into the PDE \eqref{eq:pde_vplus} and terminal condition \eqref{terminalplus} for $V^+$ using the relations
\[
V^+_{f_1}=G^+_\xi,
\quad
V^+_{f_2}=G^+-\xi G^+_\xi,\quad
f_2V^+_{f_1f_1}
=G^+_{\xi\xi},
\quad
f_2V^+_{f_1f_2}
=
-\xi G^+_{\xi\xi},
\quad
f_2V^+_{f_2f_2}
=\xi^2G^+_{\xi\xi}.
\]
\end{proof}
\begin{corollary}[Case 2 dimension reduction]
\label{cor:short_dimension_reduction}
Let \(V^-(t,f_1,f_2)\) be a classical solution of the Case~2 PDE
\eqref{eq:pde_vminus}, with terminal condition \eqref{eq:vminus_terminal}.
Then $V^-(t,f_1,f_2)=f_2G^-(t,\xi)$, where $\xi=f_1/f_2$,
and \(G^-\) solves
\begin{equation}
\label{eq:Gminus_reduced}
G^-_t
+
\frac12\left(\frac{B^-}{A^-}\right)^2
\sigma_M^2(t)\xi^2G^-_{\xi\xi}
=0,
\end{equation}
with terminal condition 
\begin{equation}
    G^-(T,\xi)=\bigl(\xi-D\bigr)^+,\label{eq:Gminus_terminal}
\end{equation}
and where
\[
A^-:=1-\varepsilon G^-_\xi,
\qquad
B^-:=1-\varepsilon\frac{G^-}{\xi}.
\]
The Margrabe variance $\sigma_M^2(t)$ is given in \eqref{eq:sigmaM_def}.
\end{corollary}
\begin{proof}
The proof is analogous to the proof of Proposition \ref{prop:long_dimension_reduction}.
\end{proof}
\remark{While Case 1 has a rollover option whose payoff is positive in the case of contango between the front and next month futures prices, the valuation of $V^+$ is related to the feedback valuation of a {\em put option} on the calendar spread \eqref{eq:Gplus_terminal}. Conversely, Case 2 has a rollover option with positive payoff in backwardation, and involves feedback valuation of a {\em call option} on the calendar spread \eqref{eq:Gminus_terminal}.}

\subsection{Explicit Solutions}\label{sec:explicit}
We can solve for $G^\pm$, and hence $V^\pm$, explicitly, up to the unique solution of a scalar equation.
\begin{proposition}[Case 1]\label{prop_explicit1}
Let \(G^+\) solve the reduced Case~1 feedback PDE \eqref{eq:Gplus_reduced} with terminal condition \eqref{eq:Gplus_terminal}. Then \(G^+(t,\xi)\) is given by the unique solution of the scalar equation
\begin{equation}\label{Gplusexplicit}
    G^+(t,\xi) = \frac{1}{1-\varepsilon}
P\!\left(t,\xi+\varepsilon G^+(t,\xi)\right),
\end{equation}
where \(P(t,\xi)\) is the standard Black-Scholes-Margrabe put price function given explicitly in \eqref{Pformula}.
\end{proposition}
\begin{proof}
The PDE \eqref{eq:Gplus_reduced} for $G^+$ can be written  
\[
G_t^+ + \frac12 \sigma_M^2(t)\left(\frac{\xi+\varepsilon G^+}
{1+\varepsilon G_\xi^+}\right)^2G_{\xi\xi}^+=0.
\]
Assume {\em a priori} that $1+\varepsilon G_\xi^+(t,\xi)>0$, so that the 
map $\xi\mapsto y:=\xi+\varepsilon G^+(t,\xi)$ is locally invertible, and define the change of variable  
\begin{equation}
    G^+(t,\xi)=H(t,y), \qquad y(t,\xi):=\xi+\varepsilon G^+(t,\xi). \label{ydef}
\end{equation}
Differentiating \(H(t,y)=G^+(t,\xi)\) gives
\[
H_y=\frac{G_\xi^+}{1+\varepsilon G_\xi^+},
\quad H_{yy}=\frac{G_{\xi\xi}^+}
{\left(1+\varepsilon G_\xi^+\right)^3}, \quad 
\xi_t=-\frac{\varepsilon G_t^+}
{1+\varepsilon G_\xi^+}.\]
Hence we obtain
\[
G_t^+
=
\left(1+\varepsilon G_\xi^+\right)H_t, \qquad 
G_{\xi\xi}^+=
\left(1+\varepsilon G_\xi^+\right)^3H_{yy}.
\]
Substitution into the reduced feedback PDE gives
\[
\left(1+\varepsilon G_\xi^+\right)
\left[
H_t
+
\frac12\sigma_M^2(t)y^2H_{yy}
\right]
=0.
\]
Since \(1+\varepsilon G_\xi^+>0\), $H$ solves the Margrabe PDE \eqref{MargPDE}.
The terminal condition is 
\[ H(T,y)=
\frac{1}{1-\varepsilon}(D-y)^+, \quad \mbox{which leads to } 
\quad H(t,y) = \frac{1}{1-\varepsilon}P(t,y), 
\]
and therefore to \eqref{Gplusexplicit}.
It remains to verify the assumed invertibility and the uniqueness of the solution to \eqref{Gplusexplicit}.
For the Margrabe put, $P_\xi(t,y)=-N(-d_1(t,y))$, so that $-1\leq P_\xi(t,y)\leq0$. 
Therefore, for $0\leq\varepsilon<1$,
\[
1-\varepsilon
\leq
1+\varepsilon G_\xi^+(t,\xi)
\leq1.
\]
Thus the invertibility condition assumed above indeed holds {\em a posteriori}. 
Finally, for fixed $(t,\xi)$ define
\[
\Phi(g)
:=
(1-\varepsilon)g
-
P(t,\xi+\varepsilon g).
\]
Since
\[
\Phi'(g)
=
1-\varepsilon
-
\varepsilon P_\xi(t,\xi+\varepsilon g)
=
1-\varepsilon\left(
1+P_\xi(t,\xi+\varepsilon g)
\right)
\geq 1-\varepsilon>0,
\]
$\Phi$ is strictly increasing. Hence \eqref{Gplusexplicit} has at most
one solution. Existence follows from the construction above.
\end{proof}
\begin{corollary}[Case 2 scalar representation]\label{cor:implicit_margrabe_minus}
For $0\leq\varepsilon<1$ and $\xi>0$, the Case~2 feedback value is determined by
\begin{equation}\label{eq:Gminus_explicit}
(1-\varepsilon)G^-(t,\xi)=C\!\left(t,\xi-\varepsilon G^-(t,\xi)\right),
\end{equation}
where $C$ is given by~\eqref{Cformula}. The solution is unique among values for which $y=\xi-\varepsilon G^-(t,\xi)>0$, and $V^-(t,f_1,f_2)=f_2G^-(t,f_1/f_2)$.
\end{corollary}
\begin{proof}
As in Case~1, set $G^-(t,\xi)=H(t,y)$, now with $y=\xi-\varepsilon G^-(t,\xi)$. The transformed PDE is the Margrabe equation and its terminal value is $(y-D)^+/(1-\varepsilon)$. Thus $H(t,y)=C(t,y)/(1-\varepsilon)$. To verify invertibility and existence directly, define
\[
\xi=y+\frac{\varepsilon}{1-\varepsilon}C(t,y).
\]
This map is strictly increasing from $(0,\infty)$ onto $(0,\infty)$ because $0\leq C_y\leq1$. It defines a unique $y$ for each $\xi>0$, and hence the stated solution. Moreover,
\[
G^-_\xi=\frac{C_y(t,y)}{1-\varepsilon+\varepsilon C_y(t,y)},\qquad
1-\varepsilon G^-_\xi=
\frac{1-\varepsilon}{1-\varepsilon+\varepsilon C_y(t,y)}
\in[1-\varepsilon,1],
\]
which verifies the nonvanishing denominator in the original PDE. At $T$, the same increasing map gives $G^-(T,\xi)=(\xi-D)^+$.
\end{proof}

The preceding transformations motivate the following direct proof of uniqueness. 

\begin{corollary}[Uniqueness of the Case 1 feedback price]
\label{cor:case1_uniqueness}
The Case~1 roll-option price consistent with the feedback relation
\eqref{eq:long_feedback_relation} and the forward-pricing convention
\eqref{eq:forward_convention} is unique. It is given by
\[
V^+(t,f_1,f_2)
=
f_2G^+\!\left(t,\frac{f_1}{f_2}\right),
\]
where \(G^+\) is the solution in Proposition~\ref{prop_explicit1}.
Consequently, this function is the unique Case~1 feedback pricing
solution of the nonlinear PDE problem
\eqref{eq:pde_vplus}--\eqref{terminalplus}.
\end{corollary}

\begin{proof}
At time \(T\), the terminal payoff in \eqref{eq:r_terminal_val} and the
Case~1 feedback relation \eqref{eq:long_feedback_relation} give
\[
V^+_T
=
\left(
D F^{\mathrm{obs}}_{T,T_2}
-
F^{\mathrm{obs}}_{T,T}
\right)^+
=
\left(
D F^{\mathrm{obs}}_{T,T_2}
-
F^{\mathrm{lic}}_{T,T}
+
\varepsilon V^+_T
\right)^+.
\]
Solving this scalar equation pointwise gives
\begin{equation}
V^+_T
=
\frac{1}{1-\varepsilon}
\left(
D F^{\mathrm{obs}}_{T,T_2}
-
F^{\mathrm{lic}}_{T,T}
\right)^+.
\label{eq:Vplus_terminal_licensed}
\end{equation}
By the forward-pricing convention \eqref{eq:forward_convention}, $V^+_t
=\mathbb{E}^{\mathbb{Q}}\left[V^+_T\mid\mathcal{F}_t\right]$.
The pair
\(
(F^{\mathrm{lic}}_{t,T},F^{\mathrm{obs}}_{t,T_2})
\)
follows the bivariate GBM dynamics in \eqref{dF2eqn}. Therefore,
using \eqref{eq:Vplus_terminal_licensed} and the Margrabe put-price
function \(P\) in \eqref{Pformula},
\[
V^+_t
=
\frac{F^{\mathrm{obs}}_{t,T_2}}{1-\varepsilon}
P\!\left(
t,
\frac{F^{\mathrm{lic}}_{t,T}}
     {F^{\mathrm{obs}}_{t,T_2}}
\right).
\]
Now suppose
\[
F^{\mathrm{obs}}_{t,T}=f_1,
\qquad
F^{\mathrm{obs}}_{t,T_2}=f_2,
\qquad
V^+_t=v.
\]
From \eqref{eq:long_feedback_relation},
\[
F^{\mathrm{lic}}_{t,T}=f_1+\varepsilon v.
\]
It follows that \(v\) must satisfy
\begin{equation}
v
=
\frac{f_2}{1-\varepsilon}
P\!\left(
t,
\frac{f_1+\varepsilon v}{f_2}
\right).
\label{eq:Vplus_implicit_full}
\end{equation}

For fixed \((t,f_1,f_2)\), define
\[
\Psi(v)
:=
v
-
\frac{f_2}{1-\varepsilon}
P\!\left(
t,
\frac{f_1+\varepsilon v}{f_2}
\right).
\]
Since \(P_\xi(t,\xi)\in[-1,0]\),
\[
\Psi'(v)
=
1-
\frac{\varepsilon}{1-\varepsilon}
P_\xi\!\left(
t,
\frac{f_1+\varepsilon v}{f_2}
\right)
\geq 1>0.
\]
Thus \eqref{eq:Vplus_implicit_full} has at most one solution for each
\((t,f_1,f_2)\). Existence is supplied by Proposition~\ref{prop_explicit1}: setting
\[
\xi=\frac{f_1}{f_2},
\qquad
v=f_2G^+(t,\xi),
\]
reduces \eqref{eq:Vplus_implicit_full} exactly to
\eqref{Gplusexplicit}. Hence $V^+(t,f_1,f_2)=f_2G^+\!\left(t,\frac{f_1}{f_2}\right)$
is the unique Case~1 feedback price. Since Proposition~\ref{prop:roll_pde}
shows that this price satisfies \eqref{eq:pde_vplus}--\eqref{terminalplus}, the result follows.
\end{proof}

\begin{corollary}[Case 2]
Let \(G^-(t,\xi)\) solve the reduced Case~2 feedback PDE \eqref{eq:Gminus_reduced} with terminal condition \eqref{eq:Gminus_terminal}. 
Then $G^-$ is determined by the unique solution of the scalar equation 
\[
G^-(t,\xi)
=
\frac{1}{1-\varepsilon}
C\!\left(t,\xi-\varepsilon G^-(t,\xi)\right),
\]
where \(C(t,\xi)\) is the standard Black-Scholes-Margrabe call price function given explicitly in \eqref{Cformula}. Furthermore, $V^-(t,f_1,f_2)=f_2\,G^-(t,\xi)$ is the {\em unique} solution in $\mathcal{C}^{1,2,2}([0,T)\times\R_+\times\R_+)$ of the nonlinear PDE problem \eqref{eq:pde_vminus}--\eqref{eq:vminus_terminal}.
\end{corollary}

\begin{proof}
The proof follows similarly to that of Proposition \ref{prop_explicit1}, with the transformation \(y=\xi-\varepsilon G^-\), and Corollary \ref{cor:case1_uniqueness}.
\end{proof}

\subsection{Oil in Negative Territory}\label{sec:bachelier}
The lognormal specification in \eqref{dF2eqn} constrains the licensed price $F_{t,T}^{\text{lic}}$ to remain positive, so that any negativity of the observed front-month futures price $F_{t,T}^{\text{obs}}$ must  come from 
the correction term $-\varepsilon V^+$ in \eqref{eq:long_feedback_relation}. We first analyze the stochastic dynamics of $F_{t,T}^{\text{obs}}$. 

Recall from \eqref{eq:fobs-dyn} that under the risk--neutral measure \(\mathbb{Q}\), $F_{t,T}^{\text{obs}}$ is a martingale: 
\begin{equation}
dF^{\text{obs}}_{t,T} \;=\; \beta_t\,dW^1_t + \gamma_t\,dW^2_t. \label{dFQ}
\end{equation}
Using the explicit solution of Proposition \ref{prop_explicit1}, we write $$\beta_t = \beta(t,F_{t,T}^{\text{obs}}, F^{\mathrm{obs}}_{t,T_2}), \qquad \gamma_t = \gamma(t,F_{t,T}^{\text{obs}}, F^{\mathrm{obs}}_{t,T_2}),$$ 
where
\begin{align*}
 \gamma(t,f_1,f_2) &= -\frac{\varepsilon\sigma_2(t)f_2}{1-\varepsilon}\,DN(-d_2(t,\xi+\varepsilon G^+(t,\xi))), \quad \xi=f_1/f_2, \\  
 \beta(t,f_1,f_2) &= 
 \sigma_1(t)f_1 -\frac{\sigma_1(t)}{\sigma_2(t)}\gamma(t,f_1,f_2),
\end{align*}
with $d_2$ defined in~\eqref{eq:d1d2}. For $0<\varepsilon<1$, $t<T$, positive integrated variance and positive $\sigma_1(t),\sigma_2(t)$, we have $\gamma_t<0$ and $\beta_t>\sigma_1(t)f_1$. This comparison concerns individual diffusion loadings. The total instantaneous variance is $\beta_t^2+2\rho\beta_t\gamma_t+\gamma_t^2$ and cannot be ordered from the first loading alone.

At $f_1=0$, the scalar representation gives $0<G^+(t,0)<\infty$, and
\[
\begin{aligned}
\gamma(t,0,f_2)&=-\frac{\varepsilon D\sigma_2(t)f_2}{1-\varepsilon}
N\!\left(-d_2(t,\varepsilon G^+(t,0))\right),\\
\beta(t,0,f_2)&=-\frac{\sigma_1(t)}{\sigma_2(t)}\gamma(t,0,f_2).
\end{aligned}
\]
Consequently, when $\sigma_M^2(t)>0$, the instantaneous variance at zero is
\[
\beta_t^2+2\rho\beta_t\gamma_t+\gamma_t^2
=\frac{\gamma_t^2}{\sigma_2(t)^2}\sigma_M^2(t)>0.
\]
Zero is therefore not a degeneracy of the diffusion in the interior of the admissible domain. Under the regularity assumptions of the model, paths can cross it into negative prices. Equivalence of the pricing and physical measures on the finite modelling horizon preserves whether such a crossing event has positive probability, although its probability need not be the same under the two measures.

\subsection{Bachelier Base Model}
In the April 2020 regime it is natural to allow the base model for $F^{\text{lic}}_{t,T}$ itself to cross zero: with Cushing storage at operational capacity at that time, the value of an undeliverable barrel was effectively negative, and the CME itself switched its options settlement model from Black (lognormal dynamics) to Bachelier (normal dynamics) on April 22, 2020 to allow for negative prices \citep{CME2020Bachelier}. The next proposition adapts the feedback mechanism to arithmetic dynamics: the scaling homogeneity of the geometric model is replaced by translation invariance, the calendar ratio by the calendar spread, and the Black-Scholes-Margrabe put and call exchange option formulas by analogous Bachelier formulas. 

\begin{proposition}[Normal dynamics: feedback PDE, spread reduction, and Bachelier limit]
\label{prop:bachelier}
Under $\mathbb{Q}$, suppose both $F^{\text{lic}}_{t,T}$ and $F^{\text{obs}}_{t,T_2}$  follow arithmetic Brownian dynamics:
\[
    dF^{\text{lic}}_{t,T} = \bar\sigma_1(t)\,dW^1_t,
    \qquad
    dF^{\text{obs}}_{t,T_2} = \bar\sigma_2(t)\,dW^2_t,
    \qquad
    d\langle W^1, W^2\rangle_t = \rho\,dt,
\]
where $\bar\sigma_1(\cdot), \bar\sigma_2(\cdot) \ge 0$ are dollar volatilities, and $F^{\text{obs}}_{t,T}$ is given by the long-side feedback relation \eqref{eq:long_feedback_relation}.  The rollover option value $V^+$ is defined in \eqref{Vpmdef}.
Then we have: 
\begin{enumerate}
    \item[(i)] Assuming $V^+(t,f_1,f_2) \in \mathcal{C}^{1,2,2}$ on $[0,T)\times \R\times\R^{(+)}$, it solves the nonlinear PDE
    \begin{equation}\label{eq:pde_bachelier}
        V^+_t
        + \frac{1}{2}\,\frac{\bar\sigma_1^2 - 2\rho\,\bar\sigma_1  C^+ + (C^+)^2}{(A^+)^2}\,V^+_{f_1 f_1}
        + \frac{\rho\,\bar\sigma_1 -  C^+}{A^+}\,\bar\sigma_2\,V^+_{f_1 f_2}
        + \frac{1}{2}\,\bar\sigma_2^2\,V^+_{f_2 f_2}
        = 0,
    \end{equation}
    with $A^+ := 1 + \varepsilon V^+_{f_1}$, $C^+ := \varepsilon\,\bar\sigma_2(t)\,V^+_{f_2}$, and terminal condition $V^+(T,f_1,f_2) = (D f_2 - f_1)^+$. 
    \item[(ii)] \emph{Spread reduction.} The PDE \eqref{eq:pde_bachelier} and its terminal data are invariant under the translation $(f_1, f_2) \mapsto (f_1 + Dc,\, f_2 + c)$, $c \in \mathbb{R}$, so $V^+$ depends on $(f_1, f_2)$ only through the calendar spread and admits the representation
    \begin{equation}
        V^+(t,f_1,f_2) = G^+(t,s), \qquad s := D f_2 - f_1 \in \mathbb{R}.\label{transfBach}
    \end{equation}
    The reduced function $G^+ : [0,T] \times \mathbb{R} \to \mathbb{R}_+$ solves the nonlinear PDE problem
    \begin{equation}\label{eq:G_bachelier}
        G^+_t + \frac{1}{2}\,\frac{\Sigma_N^2(t)}{\bigl(1 - \varepsilon G^+_s\bigr)^2}\,G^+_{ss} = 0,
        \qquad
        G^+(T,s) = s^+,
    \end{equation}
    where
    \begin{equation}\label{eq:SigmaN_def}
        \Sigma_N^2(t) := \bar\sigma_1^2(t) - 2\rho D\,\bar\sigma_1(t)\bar\sigma_2(t) + D^2\bar\sigma_2^2(t)
    \end{equation}
    is the arithmetic variance of the spread, the Bachelier analogue of $\sigma_M^2$ in \eqref{eq:sigmaM_def}. 
    \item[(iii)] \emph{Bachelier limit.} In the no-feedback limit $\varepsilon \to 0$, \eqref{eq:G_bachelier} is the heat equation and we have $V^+\to V^{(0)+}$, where $V^{(0)+}(t,f_1,f_2)=P_B(t,s)$ with     
    \begin{equation}
        P_B(t,s) = s\,N\!\left(\frac{s}{v}\right) + v\,\phi\!\left(\frac{s}{v}\right),
        \qquad
        v^2 := \int_t^T \Sigma_N^2(u)\,du, \label{PBdef}
    \end{equation}
    the Bachelier spread-option value, replacing the Black-Scholes-Margrabe formula of the lognormal model. Here $\phi$ is the standard normal p.d.f. 
\end{enumerate}
\end{proposition}

\begin{proof}
\emph{(i)} The argument follows the proof of Proposition~\ref{prop:roll_pde}. 
Writing $dF^{\text{obs}}_{t,T} = \beta_t\,dW^1_t + \gamma_t\,dW^2_t$ (zero drift by the martingale property), and matching coefficients between the It\^o expansion of $V^+$ and $dV^+ = \varepsilon^{-1}(dF^{\text{lic}} - dF^{\text{obs}})$ gives
\[
    \beta_t = \frac{\bar\sigma_1(t)}{A^+},
    \qquad
    \gamma_t = -\,\frac{\varepsilon\,V^+_{f_2}\,\bar\sigma_2(t)}{A^+},
\]
and the drift identification yields the PDE \eqref{eq:pde_bachelier} with $C^+ = \varepsilon\bar\sigma_2 V^+_{f_2}$. 
\emph{(ii)} follows from routine calculations with the transformation \eqref{transfBach}.
\emph{(iii)} at $\varepsilon = 0$, \eqref{eq:G_bachelier} is the heat equation with terminal $s^+$; the stated formula is its classical solution, $\mathbb{E}[(s + v Z)^+] = s N(s/v) + v\phi(s/v)$ for $Z \sim \mathcal{N}(0,1)$.
\end{proof}

As with the lognormal base model, we have an explicit solution, up to solving a nonlinear scalar equation. 
\begin{proposition}
The Case~1 reduced normal-dynamics feedback function, solution of the nonlinear PDE problem \eqref{eq:G_bachelier}  has explicit solution
\[
G^+(t,s)
=
\frac{1}{1-\varepsilon}
P_B\!\left(
t,
s-\varepsilon G^+(t,s)
\right),
\]
where $P_B$ is given in \eqref{PBdef}. Moreover, the solution is unique in the sense of Corollary \ref{cor:case1_uniqueness}.
\end{proposition}
\begin{proof}
The proof follows similar to those of Proposition \ref{prop_explicit1} and Corollary \ref{cor:case1_uniqueness}, and we omit the details.
\end{proof}

\begin{proposition}[Normal dynamics: Case 2 feedback PDE, spread reduction, and Bachelier limit]
\label{prop:bachelier_minus}
Under $\mathbb{Q}$, suppose both $F^{\text{lic}}_{t,T}$ and $F^{\text{obs}}_{t,T_2}$ follow arithmetic Brownian dynamics:
\[
    dF^{\text{lic}}_{t,T} = \bar\sigma_1(t)\,dW^1_t,
    \qquad
    dF^{\text{obs}}_{t,T_2} = \bar\sigma_2(t)\,dW^2_t,
    \qquad
    d\langle W^1, W^2\rangle_t = \rho\,dt,
\]
where $\bar\sigma_1(\cdot), \bar\sigma_2(\cdot) \ge 0$ are dollar volatilities, and $F^{\text{obs}}_{t,T}$ is given by the short-side feedback relation in Corollary~\ref{cor:short_side_pde}. The rollover option value $V^-$ is defined in \eqref{Vpmdef}.
Then we have:
\begin{enumerate}
    \item[(i)] Assuming $V^-(t,f_1,f_2) \in \mathcal{C}^{1,2,2}$ on $[0,T)\times \R\times\R^{(+)}$, it solves the nonlinear PDE
    \begin{equation}\label{eq:pde_bachelier_minus}
        V^-_t
        + \frac{1}{2}\,
        \frac{\bar\sigma_1^2 + 2\rho\,\bar\sigma_1 C^- + (C^-)^2}
        {(A^-)^2}\,V^-_{f_1 f_1}
        + \frac{\rho\,\bar\sigma_1 + C^-}{A^-}\,
        \bar\sigma_2\,V^-_{f_1 f_2}
        + \frac{1}{2}\,\bar\sigma_2^2\,V^-_{f_2 f_2}
        = 0,
    \end{equation}
    with
    \[
        A^- := 1-\varepsilon V^-_{f_1},
        \qquad
        C^- := \varepsilon\,\bar\sigma_2(t)\,V^-_{f_2},
    \]
    and terminal condition $V^-(T,f_1,f_2)=(f_1-Df_2)^+$.

    \item[(ii)] \emph{Spread reduction.} The PDE \eqref{eq:pde_bachelier_minus} and its terminal data are invariant under the translation
    \[
        (f_1,f_2)\mapsto(f_1+Dc,f_2+c),
        \qquad c\in\mathbb{R},
    \]
    so $V^-$ depends on $(f_1,f_2)$ only through the calendar spread and admits the representation
    \begin{equation}
        V^-(t,f_1,f_2)=G^-(t,s),
        \qquad
        s:=f_1-Df_2\in\mathbb{R}.
        \label{transfBachMinus}
    \end{equation}
    The reduced function $G^-:[0,T]\times\mathbb{R}\to\mathbb{R}_+$ solves the nonlinear PDE problem
    \begin{equation}\label{eq:G_bachelier_minus}
        G^-_t
        +\frac{1}{2}\,
        \frac{\Sigma_N^2(t)}
        {\bigl(1-\varepsilon G^-_s\bigr)^2}
        G^-_{ss}=0,
        \qquad
        G^-(T,s)=s^+,
    \end{equation}
    where $\Sigma_N^2(t)$ is given in \eqref{eq:SigmaN_def}.

    \item[(iii)] \emph{Bachelier limit.} In the no-feedback limit $\varepsilon\to0$, \eqref{eq:G_bachelier_minus} is the heat equation and we have $V^-\to V^{(0)-}$, where
    \[
        V^{(0)-}(t,f_1,f_2)=C_B(t,s)
    \]
    with
    \begin{equation}
        C_B(t,s)
        =
        s\,N\!\left(\frac{s}{v}\right)
        +
        v\,\phi\!\left(\frac{s}{v}\right),
        \qquad
        v^2:=\int_t^T\Sigma_N^2(u)\,du,
        \label{CBdef}
    \end{equation}
    the Bachelier spread-option value for the short-side payoff, replacing the Black-Scholes-Margrabe formula of the lognormal model. 
\end{enumerate}
\end{proposition}

\begin{proof}
The results follow from the analogous Case~1 arguments in Proposition~\ref{prop:bachelier}, with the signs of the feedback terms reversed and the Case~2 terminal payoff $(f_1-Df_2)^+$.
\end{proof}
Finally, we give the explicit solution. 
\begin{proposition}
\label{prop:bachelier_explicit_minus}
The Case~2 reduced normal-dynamics feedback function, solution of the nonlinear PDE problem \eqref{eq:G_bachelier_minus}, has explicit solution
\[
G^-(t,s)
=
\frac{1}{1-\varepsilon}
C_B\!\left(
t,
s-\varepsilon G^-(t,s)
\right),
\]
where $C_B$ is given in \eqref{CBdef}. Moreover, the solution is unique in the sense of the analogous Case~2 result to Corollary~\ref{cor:case1_uniqueness}.
\end{proposition}

\begin{proof}
The proof follows from the analogous Case~1 proof, with $V^+$ and $G^+$ replaced by $V^-$ and $G^-$, respectively, and $P_B$ replaced by $C_B$.
\end{proof}

\section{Properties of the Solution}\label{sec:numerics}

The representations in Proposition~\ref{prop_explicit1} and Corollary~\ref{cor:implicit_margrabe_minus} allow us to study option values, admissible prices and the effect of feedback on the diffusion. We then use Monte Carlo paths to illustrate negative-price crossings. Throughout this section, we take \(\sigma_1, \sigma_2\) to be constants.

\subsection{Feedback premium and option values}\label{sec:sufficiency}\label{sec:value}

We compare values at the same observed prices. On the common domain $\xi>0$, define the feedback premium by
\[
 \Pi^\pm(t,\xi;\ep)
 :=G^\pm(t,\xi;\ep)-G^\pm(t,\xi;0).
\]
Implicit differentiation of the scalar pricing equations gives
\begin{equation}\label{eq:dGdeps}
 \frac{\partial G^+(t,\xi;\ep)}{\partial\ep}
 =\frac{G^+(t,\xi;\ep)N(d_1^+)}{1-\ep+\ep N(-d_1^+)}\geq0,
 \qquad
 \frac{\partial G^-(t,\xi;\ep)}{\partial\ep}
 =\frac{G^-(t,\xi;\ep)N(-d_1^-)}{1-\ep+\ep N(d_1^-)}\geq0,
\end{equation}
where $d_1^\pm = d_1^\pm(t, \xi \pm \ep G^{\pm})$.
Both denominators are strictly positive. The feedback premia are therefore nonnegative and nondecreasing in $\ep$. At fixed $f_1,f_2$, the same conclusion holds for $V^\pm=f_2G^\pm$.

Figure~\ref{fig:value} shows the values and the additional premium separately. At the strike, both no-feedback values are approximately $3.10$ price units. At $\ep=0.95$, they rise to $11.25$ in Case~1 and $8.06$ in Case~2. The premium is not constant across observed prices, and its dependence on moneyness differs between the two cases.
\begin{figure}[htb]
\centering
\includegraphics[width=0.9\linewidth]{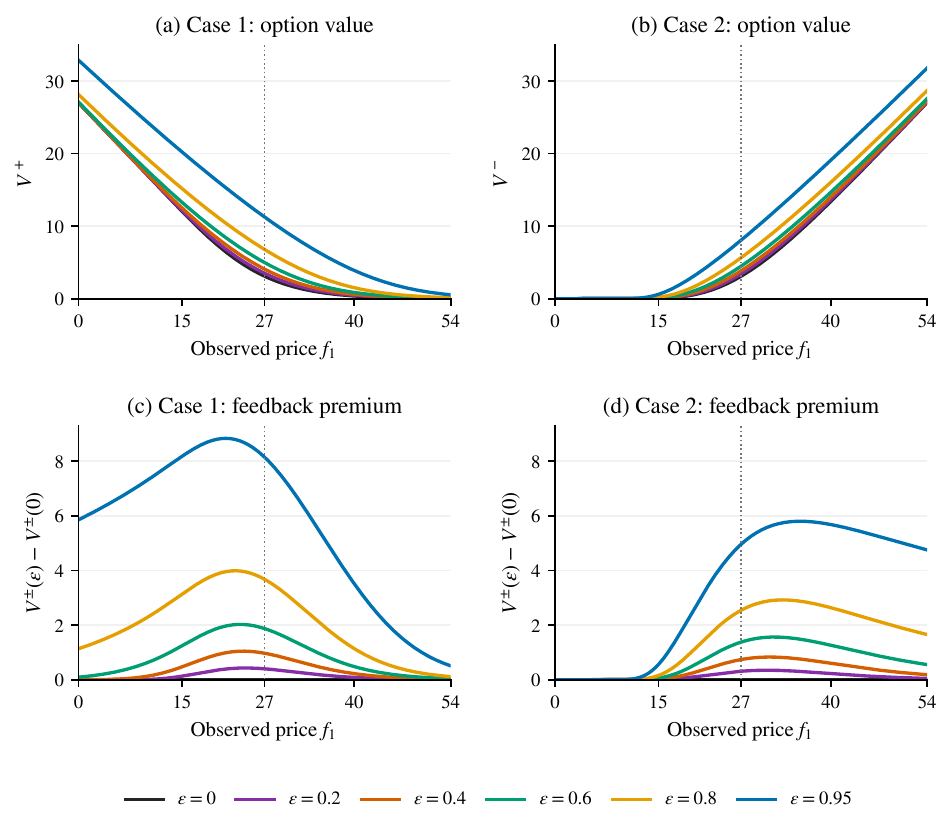}
\caption{Option values and feedback premia at fixed observed prices. Panels (a)--(b) show $V^+$ and $V^-$; panels (c)--(d) show $V^\pm(\ep)-V^\pm(0)=f_2\Pi^\pm$. Parameters are $D=0.90$, $\sigma_M=0.50$, $T-t=1/3$ year and $f_2=30$. The shared legend identifies $\ep$ in all panels; vertical dotted lines indicate $f_1=Df_2=27$. Prices and option values use the same arbitrary price unit.}
\label{fig:value}
\end{figure}

\subsection{Extension to Negative Observed Prices}
The formulas in Proposition~\ref{prop_explicit1} and Corollary~\ref{cor:implicit_margrabe_minus} were derived initially on the positive-price domain. 
Consider first Case~1. Let
$z:=\xi+\varepsilon G^+(t,\xi)$.
From Proposition~\ref{prop_explicit1}, $(1-\varepsilon)G^+(t,\xi)=P(t,z)$,
so that
\[
\xi
=
z-\frac{\varepsilon}{1-\varepsilon}P(t,z).
\]
The map from \(z>0\) to \(\xi\) is strictly increasing. As \(z\downarrow0\), using \eqref{Pformula},
$P(t,z)\to D$, 
and therefore $\xi\to-\frac{\varepsilon D}{1-\varepsilon}$.
Hence the Case~1 feedback model extends continuously from \(f_1>0\) through \(f_1=0\) and into negative observed prices. It remains well-defined on the domain
\begin{equation}\label{eq:admissible_case1}
\xi>-\frac{\varepsilon D}{1-\varepsilon},
\qquad\text{equivalently}\qquad
f_1>-\frac{\varepsilon D}{1-\varepsilon}f_2.
\end{equation}
For \(\varepsilon>0\), zero lies strictly inside this domain. Thus the observed front-month price \(F^{\text{obs}}_{t,T}\) can become negative while the licensed price \(F^{\text{lic}}_{t,T}\) remains positive. Beyond this boundary, the feedback price is no longer defined.

Case~2 is different. Writing
$z:=\xi-\varepsilon G^-(t,\xi)>0$, 
Corollary~\ref{cor:implicit_margrabe_minus} gives
$\xi=z+\varepsilon G^-(t,\xi)$.
Since \(z>0\) and \(G^-\geq0\), it follows that $\xi>0$.
Thus, in Case~2, the admissible domain remains \(f_1>0\), whereas in Case~1 the feedback model extends through zero and into negative observed prices up to the boundary \eqref{eq:admissible_case1}.

\subsection{Effective variance and zero crossing}\label{sec:where}

The variance coefficients of the reduced PDEs are
\[
 (\sigma_{\mathrm{eff}}^\pm)^2(t,\xi)
 =\left(\frac{\xi\pm\ep G^\pm}
              {1\pm\ep G^\pm_\xi}\right)^2\sigma_M^2(t).
\]
These are state-dependent variance coefficients, rather than percentage volatilities. Without feedback they reduce to $\xi^2\sigma_M^2(t)$. Figure~\ref{fig:where} shows their different behavior in the two cases.

In Case~1, when $\ep>0$, $t<T$ and $\sigma_M^2(t)>0$, we have $G^+(t,0)>0$ and $1+\ep G^+_\xi\geq1-\ep>0$. The coefficient is therefore positive at zero. Together with the diffusion calculation in Section~\ref{sec:bachelier}, this shows why zero need not prevent the observed price from crossing into negative territory. In Case~2, the admissible domain remains $\xi>0$ and the coefficient tends to zero as $\xi\downarrow0$.

\begin{figure}[htbp]
\centering
\includegraphics[width=0.9\linewidth]{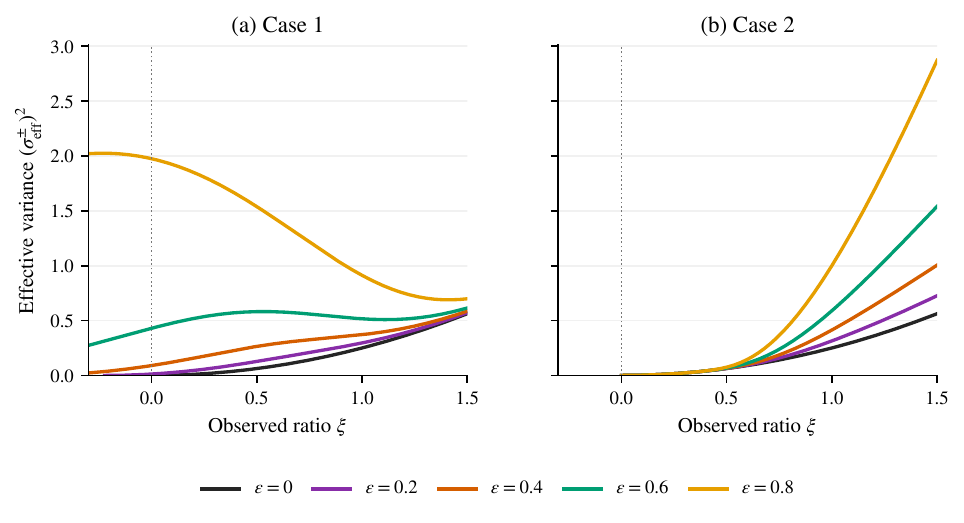}
\caption{Reduced-PDE variance coefficients in Case~1 (a) and Case~2 (b), with $D=0.90$, $\sigma_M=0.50$ and $T-t=1/3$ year. Both panels use the same scales and imbalance colors. Only admissible states are shown; vertical dotted lines indicate $\xi=0$.}
\label{fig:where}
\end{figure}

\subsection{Monte Carlo illustration}\label{sec:mc_negative}

We now illustrate crossings along paths starting from a positive observed price. Keeping the Case~1 diffusion coefficients $\beta$ and $\gamma$ derived in Section~\ref{sec:bachelier}, consider
\begin{equation}\label{eq:mc_observed_drift}
 dF_t^{\mathrm{obs}}
 =\mu F_t^{\mathrm{obs}}\,dt+\beta_t\,dW_t^1+\gamma_t\,dW_t^2,
 \qquad \mu\in\{-0.3,0,0.3\}.
\end{equation}
Time is measured in years. The second contract has zero drift and retains its lognormal dynamics, with $d\langle W^1,W^2\rangle_t=\rho\,dt$. The case $\mu=0$ is the derived risk-neutral model. The other two cases are drift-sensitivity scenarios; their drift coefficients are not estimated risk premia.
\begin{figure}[htb]
\centering
\includegraphics[width=0.9\linewidth]{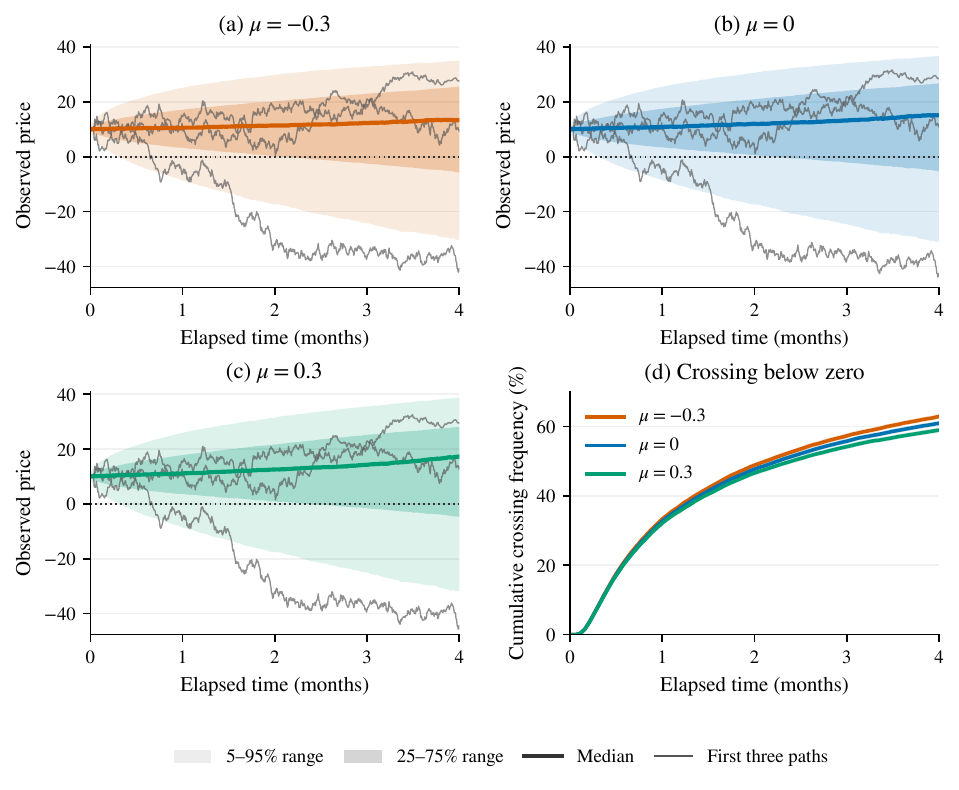}
\caption{Case~1 Monte Carlo illustration. Panels (a)--(c) show
the median, central 50\% and 90\% ranges, and the first
three sample paths for $\mu=-0.3,0,0.3$, respectively. Panel (d) shows cumulative crossing frequencies with pointwise
95\% intervals.
Inputs are $F_{0,T}^{\mathrm{obs}}=10$, $F_{0,T_2}=25$,
$\ep=0.80$, $D=1$, $T=1/3$ year and
$(\sigma_1,\sigma_2,\rho)=(0.50,0.30,0.30)$,
giving $\sigma_M=0.50$.}
\label{fig:mc_negative}
\end{figure}

For the inputs in Figure~\ref{fig:mc_negative}, the fractions of the paths detected below zero at least once are $62.9\%$, $61.0\%$ and $59.0\%$ for $\mu=-0.3$, $0$ and $0.3$, respectively. Their pointwise 95\% Monte Carlo intervals at the final time are approximately $[62.3,63.6]\%$, $[60.3,61.6]\%$ and $[58.3,59.7]\%$. These cumulative frequencies record whether a path has been detected below zero at any simulation time, rather than whether it is negative at maturity.

\section{Applications to the WTI Collapse and the Nickel Squeeze}\label{sec:calibrations}

We use our model pricing formulas in Section~\ref{sec:explicit} to calculate the implied imbalance given the observed futures price \(F^{\text{obs}}_{t,T}\) and a posited benchmark price \(F^{\text{lic}}_{t,T}\) which could plausibly have prevailed in the absence of exit pressure.

For a downward distortion (Case~1) or an upward distortion (Case~2), respectively, the implied imbalance satisfies
\begin{equation*}\label{eq:implied_eps}
 \ep_t^+V^+(t,F^{\text{obs}}_{t,T},F^{\text{obs}}_{t, T_2};\ep_t)=F^{\text{lic}}_{t,T}-F^{\text{obs}}_{t,T},
 \qquad
 \ep_t^-V^-(t,F^{\text{obs}}_{t,T},F^{\text{obs}}_{t, T_2};\ep_t)=F^{\text{obs}}_{t,T} - F^{\text{lic}}_{t,T}.
\end{equation*}
Using the scalar pricing identities, with $V^\pm$ evaluated at the observed price pair and the corresponding imbalance, we obtain
\begin{align*}
    (1-\ep_t^+)V^+ &= F^{\text{obs}}_{t,T_2} P\left(t,\frac{F^{\text{obs}}_{t,T} + \ep_t V^+}{F^{\text{obs}}_{t,T_2}}\right) = F^{\text{obs}}_{t,T_2} P\left(t,\frac{F^{\text{lic}}_{t,T}}{F^{\text{obs}}_{t,T_2}}\right),\\
    (1-\ep_t^-)V^- &= F^{\text{obs}}_{t,T_2}C\left(t,\frac{F^{\text{obs}}_{t,T} - \ep_t V^-}{F^{\text{obs}}_{t,T_2}}\right)= F^{\text{obs}}_{t,T_2}C\left(t,\frac{F^{\text{lic}}_{t,T}}{F^{\text{obs}}_{t,T_2}}\right).
\end{align*}
Solving for the imbalance gives
\begin{equation}\label{eq:nickel_scenario_inversion}
 \ep_t^+=
 \frac{F^{\text{lic}}_{t,T} - F^{\text{obs}}_{t,T}}
 {F^{\text{lic}}_{t,T} - F^{\text{obs}}_{t,T} + F^{\text{obs}}_{t,T_2} P\left(t,\frac{F^{\text{lic}}_{t,T}}{F^{\text{obs}}_{t,T_2}}\right)}, \qquad
 \ep_t^-=
 \frac{F^{\text{obs}}_{t,T} - F^{\text{lic}}_{t,T}}
 {F^{\text{obs}}_{t,T} - F^{\text{lic}}_{t,T} + F^{\text{obs}}_{t,T_2} C\left(t,\frac{F^{\text{lic}}_{t,T}}{F^{\text{obs}}_{t,T_2}}\right)}.
\end{equation}

For the lognormal formulas, we require $F^{\mathrm{lic}}_{t,T}>0$
and $F^{\mathrm{obs}}_{t,T_2}>0$. Case~1 applies when
$F^{\mathrm{obs}}_{t,T}\leq F^{\mathrm{lic}}_{t,T}$, and Case~2
when $F^{\mathrm{obs}}_{t,T}\geq F^{\mathrm{lic}}_{t,T}$.
With positive integrated variance, the corresponding option values
are strictly positive, so each applicable inversion gives a unique
$\ep_t^\pm\in[0,1)$. Each time point is a separate constant imbalance valuation; the sequence does not specify dynamics for $\ep_t$. 
We use $D=1$ in both baseline event illustrations.

\subsection{WTI: the April 2020 collapse}\label{sec:april2020_calibration}\label{sec:implied}

On April~20, 2020, the May and June WTI contracts recorded official daily settlement prices of $-37.63$ and $20.43$ USD/bbl, respectively. The May contract expired the next day \citep{CFTC2020WTI}. For the settlement inputs in Table~\ref{tab:wti_calibration}, the implied imbalance is approximately $\ep=0.87$.

\begin{table}[!htbp]
\centering
\begin{tabular}{lr}
\toprule
Input & WTI: Case~1\\
\midrule
Observed front price $F^{\text{obs}}_{t,T}$ & $-37.63$\\
Second-contract price $F^{\text{obs}}_{t,T_2}$ & $20.43$\\
Benchmark $F^{\text{lic}}_{t,T}$ & $13.14$\\
Benchmark ratio $\tfrac{F^{\text{lic}}_{t,T}}{F^{\text{obs}}_{t,T_2}}$ & $0.64$\\
Discount factor $D$ & $1.00$\\
Effective volatility $\sigma_M$ & $0.52$\\
Remaining time $T-t$ (years) & $2/252$\\
\midrule
Implied imbalance $\ep$ & $0.87$\\
\bottomrule
\end{tabular}
\caption{Conditional WTI calculation using the May and June official
settlements on April~20, 2020. Prices are in USD/bbl. The benchmark
is assumed, and volatility is estimated from pre-event returns.}
\label{tab:wti_calibration}
\end{table}

We next repeat the calculation intraday on April~20 using 268
timestamp matched five-minute closing prices for the May and June
contracts. These provide the observed pair $(F^{\mathrm{obs}}_{t,T},F^{\mathrm{obs}}_{t,T_2})$ at each sampled
time. The five-minute May and June WTI futures prices used in the following are obtained from Bloomberg.

To specify the benchmark throughout the day, we set
\[
F^{\mathrm{lic}}_{t,T}
=\lambda F^{\mathrm{obs}}_{t,T_2},
\qquad
\lambda
=
\frac{F^{\mathrm{obs}}_{t_0,T}}
     {F^{\mathrm{obs}}_{t_0,T_2}}
\approx 0.64,
\]
where $t_0$ is the first observation, midnight on April~20.
Thus the benchmark follows the June contract while preserving the initial May--June price ratio. This construction sets the initial implied imbalance to zero. It measures subsequent departures from the initial relation, without establishing that the market was free of exit pressure at $t_0$.
\begin{figure}[htbp]
\centering
\includegraphics[width=\linewidth]
{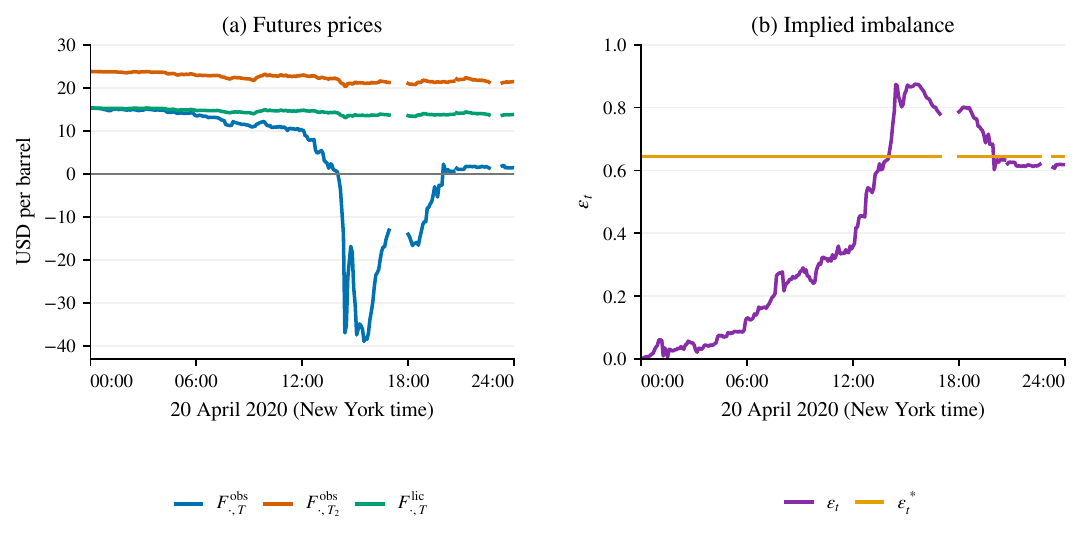}
\caption{WTI prices and conditional implied imbalance on
April~20, 2020. Panel~(a) shows the May contract
$F^{\mathrm{obs}}_{t,T}$, the June contract
$F^{\mathrm{obs}}_{t,T_2}$ and the assumed benchmark
$F^{\mathrm{lic}}_{t,T}=0.64\times F^{\mathrm{obs}}_{t,T_2}$.
Panel~(b) shows the implied imbalance $\ep_t^+$ in purple
and the zero-price threshold $\ep_t^*$ in orange.}
\label{fig:implied}
\end{figure}


The volatility input is estimated from the difference between the May and June log returns over March~18--April~17. We retain 5,974 returns between consecutive five-minute observations, excluding returns across missing bars or session breaks. Annualizing with $252\times276$ five-minute intervals gives $\sigma_M=0.5203$. This is a historical relative-price volatility used as a pricing input, rather than an option implied volatility. The remaining horizon decreases to April~21 at 14:30 and is measured in the same trading time units, excluding the daily maintenance interval.

At each timestamp, we first construct $F^{\mathrm{lic}}_{t,T}$ from the observed June price. We then evaluate the Margrabe put at $F^{\mathrm{lic}}_{t,T}/F^{\mathrm{obs}}_{t,T_2}=\lambda$, using the remaining horizon, and substitute its value and the observed price gap into~\eqref{eq:nickel_scenario_inversion}. This gives $\ep_t^+$ directly. The prices in panel~(a) of Figure~\ref{fig:implied} therefore determine the implied-imbalance series in panel~(b), conditional on the benchmark and pricing inputs.

The May price falls sharply while the June price, and hence the assumed benchmark, remain positive. The resulting increase in the price gap raises the implied imbalance, which reaches $0.87$ at 15:30, when the sampled May price is $-\$38.96$/bbl. As the May price recovers, the inferred imbalance declines to approximately $0.6182$ at the final observation. The peak estimate differs from the settlement estimate because the two calculations use different observed price pairs.

\subsection{Nickel: the March 2022 squeeze}\label{sec:nickel_calibration}\label{sec:nickel_data}
The three-month nickel price reached $\$101{,}365$/tonne at 06:08 on March~8, 2022, before retreating; trading was suspended at 08:15 \citep{OliverWyman2023}. The LME subsequently canceled trades executed from midnight that day, so the high is a reported pre-cancellation trade price, not a settlement \citep{LME2022marketupdate}.

As in the WTI case, we specify a benchmark price and calculate the imbalance implied by \eqref{eq:nickel_scenario_inversion}. 
See Table \ref{tab:nickel_calibration}.
\begin{table}[htb]
\centering
\begin{tabular}{lr}
\toprule
Input & Nickel: Case~2\\
\midrule
Reported three-month high $F^{\mathrm{obs}}_{t,T}$
    & $101{,}365$\\
Fifteen-month price $F^{\mathrm{obs}}_{t,T_2}$
    & $60{,}195$\\
Benchmark $F^{\mathrm{lic}}_{t,T}$
    & $62{,}863.41$\\
Benchmark ratio $\tfrac{F^{\text{lic}}_{t,T}}{F^{\text{obs}}_{t,T_2}}$ & $1.04$\\
Discount factor $D$ & $1.00$\\
Effective volatility $\sigma_M$ & $0.5$\\
Remaining time $T-t$ (years) & $1/4$\\
\midrule
Implied imbalance $\ep$ & $0.9346$\\
\bottomrule
\end{tabular}
\caption{Conditional nickel calculation at the reported
three-month high. Prices are in USD/tonne. The fifteen-month
price and benchmark are scenario values; the ratio and effective
volatility are approximate inputs extracted from the historical
chart.}
\label{tab:nickel_calibration}
\end{table}
Here, however, our dataset does not contain synchronized intraday three-month and fifteen-month quotes. We therefore retain a constructed path through the reported three-month levels of approximately $\$50{,}000$ at 01:00, $\$60{,}000$ at 04:49, $\$70{,}000$ at 05:44 and $\$101{,}365$ at 06:08 \citep[p.~21]{OliverWyman2023}. Times are London time (GMT). The fifteenth-month path is illustrative; it is not an observed contract series. Intraday data is not available through Bloomberg for the period March 8 -- March 15, 2022.

To specify the benchmark, we use the three-month and fifteen-month curves in Exhibit~5 of \citet[p.~15]{OliverWyman2023}. The final plotted values give an approximate ratio
\[
    \lambda = \frac{F_{3M}}{F_{15M}} \approx 1.04.
\]
We retain this ratio throughout the illustration and set $F^{\mathrm{lic}}_{t,T}=\lambda F^{\mathrm{obs}}_{t,T_2}$.

At the reported high, the observed price exceeds the benchmark
by approximately $38{,}501.59$. Equation~\eqref{eq:nickel_scenario_inversion} therefore gives
$\ep\approx0.9346$. This value depends on the assumed
fifteen-month path and benchmark; it is not an independent
estimate of historical position imbalance.
\begin{figure}[htb]
\centering
\includegraphics[width=\linewidth]
{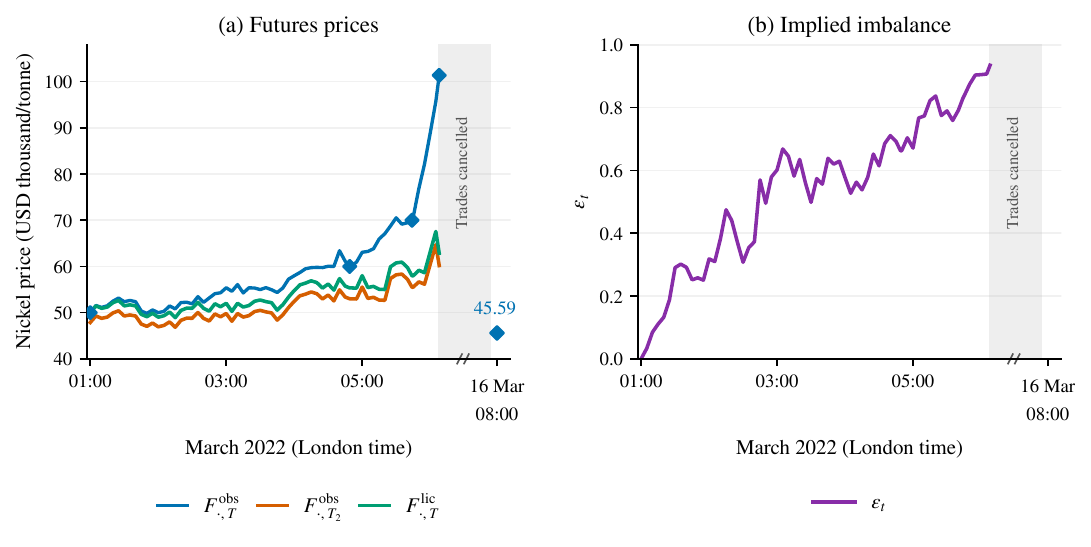}
\caption{Nickel price scenario and conditional Case~2 imbalance.
Panel (a) shows the constructed three-month path
$F^{\mathrm{obs}}_{\cdot,T}$, the illustrative fifteen-month path
$F^{\mathrm{obs}}_{\cdot,T_2}$ and the posited benchmark
$F^{\mathrm{lic}}_{\cdot,T}$. Diamonds on March~8 mark reported
price anchors. Panel (b) shows the resulting implied imbalance. 
The displayed March~8 trades were subsequently canceled.
The isolated blue diamond is the three-month opening price of
$\$45{,}590$/tonne at 08:00 on March~16
\citep{OliverWyman2023}.}
\label{fig:nickel}
\end{figure}

\section{Conclusion}\label{sec:conclusion}
This paper develops an explanatory model of how delivery constraints can amplify commodity futures price movements. Separating participants by their ability to participate in physical delivery links the direction of the price distortion to the imbalance of constrained positions. Excess constrained longs generate downward pressure, while excess constrained shorts generate upward pressure. A roll option connects this pressure to the observed price, producing a nonlinear feedback relation.

The main analytical contribution is an exact representation of the resulting pricing solution. Under lognormal benchmark dynamics, dimension reduction and a change of variables connect the nonlinear PDE to the Margrabe exchange option formula. At a given observed price state, valuation requires one scalar root. The normal dynamics extension provides an analogous Bachelier representation. These results make the feedback mechanism accessible without numerically solving the PDE and allow its effects on option values, admissible prices and effective variance to be studied directly. In particular, the lognormal benchmark can remain positive while the observed price crosses zero.

The WTI and nickel illustrations show how the mechanism can generate large distortions in either direction. Their role is to relate the model to recognizable market events and to quantify the imbalance required under stated assumptions. They do not establish that feedback alone caused either event, or that the implied imbalance equals the actual historical position imbalance. Benchmark prices, the later-maturity futures price and the relevant horizon all matter. The model complements existing explanations based on fundamentals, liquidity and market structure by isolating the consequences of constrained position closure.

Several extensions follow naturally. The first is a nonconstant imbalance. A deterministic time-dependent $\varepsilon(t)$ could represent a changing concentration of positions as a closure date approaches. A stochastic $\varepsilon_t$ could respond to prices, margin calls, liquidity and changes in participants' positions. Such extensions require a new dynamic pricing derivation: differentiating the feedback relation introduces terms involving changes in $\varepsilon_t$ and, in the stochastic case, covariation terms. Substituting a time-varying estimate into the constant-$\varepsilon$ formula is therefore not, by itself, a solution of the extended model. It would be useful to determine which reductions remain available and how stochastic imbalance changes the probability and persistence of extreme prices.

A second direction is empirical. Applying the framework to a broader set of historical squeezes, liquidation episodes and less stressed contract expiries would help assess when this mechanism is informative and when it is insufficient. This would require synchronized prices across maturities, measures of liquidity and delivery conditions, and sensitivity analysis for the fundamental benchmark. Comparisons across events could then examine whether the model's implications extend beyond the two motivating examples.

The most important data requirement is contract-specific information on positions and delivery capability. Aggregate open interest or broad commercial and noncommercial classifications do not identify the licensed and unlicensed groups used here. Positions would need to be linked, where possible, to access to eligible inventory, storage and delivery arrangements. Such data would allow the imbalance to be measured independently of the price inversion, turning the present explanatory illustrations into tests of the proposed feedback mechanism.

\appendix
\section{Contract specifications and delivery constraints}\label{app:contracts}

This appendix summarizes the contract features and delivery constraints relevant to the two applications. Throughout the paper, ``licensed'' refers to the ability to use the relevant delivery arrangements. It is a modelling distinction, rather than a regulatory classification, and
does not necessarily coincide with the distinction between commercial and non-commercial traders. The tables summarize selected features of the cited contract specifications.

\subsection{Case 1: NYMEX WTI crude oil}\label{appendix:contract_details}
NYMEX WTI futures are listed by delivery month and provide for physical
delivery of eligible crude at Cushing, Oklahoma. A trader holding a long position into the delivery process must be able to arrange receipt of the oil. Table~\ref{tab:wti_specs} summarizes the relevant contract features.
\begin{table}[H]
\centering\small
\renewcommand{\arraystretch}{1.2}
\begin{tabular}{@{}>{\raggedright\arraybackslash}p{3.8cm}>{\raggedright\arraybackslash}p{10cm}@{}}
\toprule
\textbf{Feature} & \textbf{NYMEX WTI futures}\\
\midrule
Contract and code & Light Sweet Crude Oil futures; CL.\\
Contract unit & 1,000 US barrels.\\
Price quotation & US dollars per barrel.\\
Delivery calendar & Monthly delivery contracts.\\
End of trading & The third business day before the 25th calendar day of the preceding month; if the 25th is not a business day, the third business day before the last business day preceding it, subject to exchange calendar rules.\\
Settlement & Physical delivery during the contract's delivery month.\\
Eligible material & Light sweet crude meeting the exchange's quality and delivery requirements.\\
Delivery arrangements & Eligible pipeline or storage facilities at Cushing, Oklahoma.\\
Model-relevant constraint & A long without access to the necessary receipt, storage and operational arrangements may need to close before delivery.\\
\bottomrule
\end{tabular}
\caption{Selected WTI contract features relevant to Case~1. Source: NYMEX Rulebook, Chapter~200 \citep{CMERulebook200}.}
\label{tab:wti_specs}
\end{table}

\subsection{Case 2: LME nickel}\label{app:lme_nickel_specs}

LME nickel futures trade for specified delivery dates, known as prompt dates, including the widely quoted 
three-month tenor. Each position has a fixed prompt date; maintaining a constant three-month exposure requires rolling the position. Table~\ref{tab:lme_nickel_specs} summarizes the relevant contract features.
\begin{table}[H]
\centering\small
\renewcommand{\arraystretch}{1.2}
\begin{tabular}{@{}>{\raggedright\arraybackslash}p{3.8cm}>{\raggedright\arraybackslash}p{10cm}@{}}
\toprule
\textbf{Feature} & \textbf{LME nickel futures}\\
\midrule
Contract and code & Primary nickel futures; NI.\\
Contract unit & 6 metric tonnes.\\
Price quotation & US dollars per metric tonne.\\
Delivery calendar & Daily prompt dates out to three months, weekly dates from three to six months, and monthly dates thereafter out to 63 months.\\
End of trading & Up to the close of the first Ring on the business day before the prompt date, under the exchange's prompt-date rules.\\
Settlement & Physical delivery on the specified prompt date.\\
Eligible material & Approved nickel meeting the minimum 99.80\% purity requirement and the exchange's brand and form specifications.\\
Delivery arrangements & Transfer of warrants for eligible metal in LME-approved warehouses.\\
Model-relevant constraint & A short without access to eligible metal may need to buy it or close the futures position; ownership of other nickel products does not itself meet the delivery requirement.\\
\bottomrule
\end{tabular}
\caption{Selected LME nickel contract features relevant to Case~2.
Sources: LME contract specifications \citep{LMENickelSpecs} and
prompt-date guidance \citep{LMEPromptDates}.}
\label{tab:lme_nickel_specs}
\end{table}

\setlength{\emergencystretch}{2em}
\renewcommand{\bibfont}{\small}
\setlength{\bibsep}{6pt plus 1pt}
\bibliographystyle{abbrvnat}
\bibliography{references}

\begin{thebibliography}{23}
\providecommand{\natexlab}[1]{#1}
\providecommand{\url}[1]{\texttt{#1}}
\expandafter\ifx\csname urlstyle\endcsname\relax
  \providecommand{\doi}[1]{doi: #1}\else
  \providecommand{\doi}{doi: \begingroup \urlstyle{rm}\Url}\fi

\bibitem[Basak and Pavlova(2016)]{pavlova2016model}
S.~Basak and A.~Pavlova.
\newblock A model of financialization of commodities.
\newblock \emph{The Journal of Finance}, 71\penalty0 (4):\penalty0 1511--1556,
  2016.
\newblock \doi{10.1111/jofi.12408}.

\bibitem[Black(1976)]{black1976futures}
F.~Black.
\newblock The pricing of commodity contracts.
\newblock \emph{Journal of Financial Economics}, 3\penalty0 (1--2):\penalty0
  167--179, 1976.
\newblock \doi{10.1016/0304-405X(76)90024-6}.

\bibitem[Bouchouev(2020)]{Bouchouev2020negative}
I.~Bouchouev.
\newblock Negative oil prices put spotlight on investors.
\newblock Energy Risk, Apr. 2020.
\newblock URL
  \url{https://www.energyrisk.com/oil-refined-products/7949778/negative-oil-prices-put-spotlight-on-investors}.
\newblock Published 30 April 2020.

\bibitem[Bouchouev(2023)]{Bouchouev2023virtualbarrels}
I.~Bouchouev.
\newblock \emph{Virtual Barrels: Quantitative Trading in the Oil Market}.
\newblock Springer Texts in Business and Economics. Springer Nature, Cham,
  Switzerland, 2023.
\newblock ISBN 978-3-031-36150-0.
\newblock \doi{10.1007/978-3-031-36151-7}.

\bibitem[Chan et~al.(2015)Chan, Sircar, and Stein]{Sircar2015feedback}
P.~Chan, R.~Sircar, and M.~V. Stein.
\newblock A feedback model for the financialization of commodity markets.
\newblock \emph{SIAM Journal on Financial Mathematics}, 6\penalty0
  (1):\penalty0 870--899, 2015.
\newblock \doi{10.1137/140995349}.

\bibitem[Cheng and Xiong(2014)]{Cheng2014}
I.-H. Cheng and W.~Xiong.
\newblock Financialization of commodity markets.
\newblock \emph{Annual Review of Financial Economics}, 6:\penalty0 419--441,
  2014.
\newblock \doi{10.1146/annurev-financial-110613-034432}.

\bibitem[{CME Clearing}(2020)]{CME2020Bachelier}
{CME Clearing}.
\newblock Switch to {Bachelier} options pricing model: Effective {April} 22,
  2020, Apr. 2020.
\newblock URL
  \url{https://www.cmegroup.com/notices/clearing/2020/04/Chadv20-171.html}.
\newblock Clearing Advisory 20-171, issued 21 April 2020.

\bibitem[{CME Group}(n.d.)]{CMERulebook200}
{CME Group}.
\newblock \emph{{NYMEX} Rulebook, Chapter 200: Light Sweet Crude Oil Futures}.
\newblock CME Group, n.d.
\newblock URL \url{https://www.cmegroup.com/rulebook/NYMEX/2/200.pdf}.
\newblock Accessed 30 September 2026.

\bibitem[{Commodity Futures Trading Commission}(2020)]{CFTC2020WTI}
{Commodity Futures Trading Commission}.
\newblock Interim staff report on trading in {NYMEX WTI} crude oil futures
  contract leading up to, on, and around {April} 20, 2020.
\newblock Technical report, Commodity Futures Trading Commission, Nov. 2020.
\newblock URL
  \url{https://www.cftc.gov/media/5296/InterimStaffReportNYMEX_WTICrudeOil/download}.

\bibitem[{Financial Conduct Authority}(2023)]{FCA2023LME}
{Financial Conduct Authority}.
\newblock Update on our public statement on the {London Metal Exchange}, Mar.
  2023.
\newblock URL
  \url{https://www.fca.org.uk/news/statements/update-our-public-statement-london-metal-exchange}.
\newblock Published 3 March 2023.

\bibitem[Frey and Stremme(1997)]{FreyStremme1997}
R.~Frey and A.~Stremme.
\newblock Market volatility and feedback effects from dynamic hedging.
\newblock \emph{Mathematical Finance}, 7\penalty0 (4):\penalty0 351--374, 1997.
\newblock \doi{10.1111/1467-9965.00036}.

\bibitem[{High Court of Justice}(2023)]{Elliott2023}
{High Court of Justice}.
\newblock {R (Elliott Associates L.P. and others) v London Metal Exchange and
  LME Clear Limited}.
\newblock [2023] EWHC 2969 (Admin), 2023.
\newblock URL
  \url{https://www.lme.com/-/media/files/news/co19952022-ac2022lon000890-co20072022-ac2022lon000901-final-judgment-291123.pdf}.
\newblock King's Bench Division, Administrative Court. Judgment dated 29
  November 2023.

\bibitem[{London Metal Exchange}(2022)]{LME2022marketupdate}
{London Metal Exchange}.
\newblock Nickel market update, 2022.
\newblock URL
  \url{https://www.lme.com/-/media/Files/News/Notices/2022/03/TRADING-22-057-NICKEL-MARKET-UPDATE_.pdf}.
\newblock Notice 22/057, 10 March 2022.

\bibitem[{London Metal Exchange}(n.d.{\natexlab{a}})]{LMENickelSpecs}
{London Metal Exchange}.
\newblock {LME} nickel: Contract specifications, n.d.{\natexlab{a}}.
\newblock URL
  \url{https://www.lme.com/en/Metals/Non-ferrous/LME-Nickel/Contract-specifications}.
\newblock Accessed 28 September 2026; cited for selected contract features, not
  as an archived 2022 rulebook.

\bibitem[{London Metal Exchange}(n.d.{\natexlab{b}})]{LMEPromptDates}
{London Metal Exchange}.
\newblock Prompt date structure, n.d.{\natexlab{b}}.
\newblock URL
  \url{https://www.lme.com/Sustainability-and-Physical-Markets/Physical-market-benefits/Prompt-date-structure}.
\newblock Accessed 28 September 2026.

\bibitem[Ma and Wang(2021)]{Ma2021strategic}
C.~Ma and X.~Wang.
\newblock Strategic interactions and negative oil prices.
\newblock \emph{Annals of Financial Economics}, 16\penalty0 (3):\penalty0
  2150013, 2021.
\newblock \doi{10.1142/S2010495221500135}.

\bibitem[Margrabe(1978)]{Margrabe1978}
W.~Margrabe.
\newblock The value of an option to exchange one asset for another.
\newblock \emph{The Journal of Finance}, 33\penalty0 (1):\penalty0 177--186,
  1978.
\newblock \doi{10.1111/j.1540-6261.1978.tb03397.x}.

\bibitem[{Oliver Wyman}(2023)]{OliverWyman2023}
{Oliver Wyman}.
\newblock Independent review of events in the nickel market in {March} 2022:
  Final report.
\newblock Independent review, Oliver Wyman and {NERA}, commissioned by {LME
  Group}, Jan. 2023.
\newblock URL
  \url{https://www.lme.com/-/media/Files/Trading/New-initiatives/Nickel-independent-review/Independent-Review-of-Events-in-the-Nickel-Market-in-March-2022---Final-Report.pdf}.

\bibitem[Schwartz(1997)]{Schwartz1997}
E.~S. Schwartz.
\newblock The stochastic behavior of commodity prices: Implications for
  valuation and hedging.
\newblock \emph{The Journal of Finance}, 52\penalty0 (3):\penalty0 923--973,
  1997.
\newblock \doi{10.1111/j.1540-6261.1997.tb02721.x}.

\bibitem[Sircar and Papanicolaou(1998)]{SircarPapanicolaou1998}
K.~R. Sircar and G.~Papanicolaou.
\newblock General {Black--Scholes} models accounting for increased market
  volatility from hedging strategies.
\newblock \emph{Applied Mathematical Finance}, 5\penalty0 (1):\penalty0 45--82,
  1998.
\newblock \doi{10.1080/135048698334727}.

\bibitem[Sockin and Xiong(2015)]{SockinXiong2015}
M.~Sockin and W.~Xiong.
\newblock Informational frictions and commodity markets.
\newblock \emph{The Journal of Finance}, 70\penalty0 (5):\penalty0 2063--2098,
  2015.
\newblock \doi{10.1111/jofi.12261}.

\bibitem[Tang and Xiong(2012)]{Tang2012index}
K.~Tang and W.~Xiong.
\newblock Index investment and the financialization of commodities.
\newblock \emph{Financial Analysts Journal}, 68\penalty0 (6):\penalty0 54--74,
  2012.
\newblock \doi{10.2469/faj.v68.n6.5}.

\bibitem[{U.S. Energy Information Administration}(2020)]{EIAcushing}
{U.S. Energy Information Administration}.
\newblock Cushing, {OK} ending stocks excluding {SPR} of crude oil.
\newblock Weekly data series \texttt{W\_EPC0\_SAX\_YCUOK\_MBBL}, 2020.
\newblock URL
  \url{https://www.eia.gov/dnav/pet/hist/LeafHandler.ashx?n=PET&s=W_EPC0_SAX_YCUOK_MBBL&f=W}.
\newblock 2020 observations; accessed 30 September 2026.

\end{thebibliography}

\end{document}